\documentclass[12pt,reqno]{amsart}
\usepackage[margin=1in]{geometry}

\usepackage[T1]{fontenc}
\usepackage{tgtermes}
\usepackage[american]{babel}
\usepackage{amsmath,amssymb,amsthm}
\usepackage{mathrsfs}
\usepackage{graphicx}

\usepackage{caption}\usepackage{subcaption}\usepackage{subfig}
\usepackage[all]{xy}

\usepackage[usenames,dvipsnames,svgnames,table]{xcolor}
\usepackage{tikz}
\usetikzlibrary{matrix,arrows,backgrounds,calc,chains,automata,positioning,patterns}
\usetikzlibrary{shapes.geometric,calc}
\usepackage{enumerate}

\newtheorem{theorem}{Theorem}[section]
\newtheorem{proposition}[theorem]{Proposition}
\newtheorem{corollary}[theorem]{Corollary}
\newtheorem{lemma}[theorem]{Lemma}
\theoremstyle{definition}
\newtheorem{definition}{Definition}[section]

\newtheorem{example}{Example}

\newtheorem{remark}{Remark}

\newtheorem*{thm*}{Theorem}
\newtheorem*{lem*}{Lemma}
\newtheorem*{prop*}{Proposition}

\newcommand{\set}[1]{\left\{#1\right\}}
\newcommand{\abs}[1]{\left|#1\right|}
\newcommand{\paren}[1]{\left(#1\right)}

\newcommand{\R}{\mathbb{R}}

\newcommand{\eng}{\mathscr{E}}
\newcommand{\Eng}{{\breve{\mathscr{E}}}}

\renewcommand{\Eng}{\eng}

\newcommand{\edge}{E}

\newcommand{\trace}{\operatorname{Tr}}

\newcommand{\eps}{\varepsilon}

\renewcommand{\graph}{\mathsf{G}}
\newcommand{\gmet}{\graph^{\text{met}}}
\newcommand{\res}{r}

\newcommand{\dom}{\operatorname{dom}}
\renewcommand{\dom}{\operatorname{Dom}}

\newcommand{\Dom}{\operatorname{Dom}}
\newcommand{\F}{\mathcal{F}}

\newcommand{\Xm}{X_{\Am}}
\newcommand{\Gmet}{G^{\text{met}}}
\newcommand{\haus}{\mathcal{H}^1}

\DeclareMathOperator{\const}{const}

\DeclareMathOperator{\Diam}{Diam}
\DeclareMathOperator{\diam}{diam}

\newcommand{\Ka}{K_{\alpha}}

\newcommand{\A}{\mathcal{A}}
\newcommand{\An}{\mathcal{A}^n}
\newcommand{\Am}{\mathcal{A}^m}

\newcommand{\N}{\mathbb{N}}
\newcommand{\supp}{supp}
\providecommand{\abs}[1]{\left\vert#1\right\vert}
\providecommand{\norm}[1]{\left\Vert#1\right\Vert}

\def\s-s{self-similar}
\numberwithin{equation}{section}
\title[Hanoi-type fractal quantum graphs]{Energy and Laplacian on Hanoi-type fractal quantum graphs}
\author[P.~Alonso-Ruiz]{Patricia Alonso-Ruiz}
\address[P.~Alonso-Ruiz]{Institute of Stochastics, Ulm University, Ulm, Germany}
\email[P.~Alonso-Ruiz]{patricia.alonso@uni-ulm.de}
\urladdr{www.uni-ulm.de/en/mawi/institute-of-stochastics/mitarbeiter/patricia-alonso-ruiz.html}

\author[D.~J.~Kelleher]{Daniel J.~Kelleher}
\address[D.~J.~Kelleher]{Department of Mathematics, Purdue University
150 N. University Street, West Lafayette, IN 47907-2067, USA}
\email[D.~J.~Kelleher]{dkellehe@purdue.edu}
\urladdr{http://web.ics.purdue.edu/~dkellehe/}

\author[A.~Teplyaev]{Alexander Teplyaev}
\address[A.~Teplyaev]{Department of Mathematics, University of Connecticut, Storrs, CT 06269-3009, USA}
\email[A.~Teplyaev]{teplyaev@uconn.edu}
\urladdr{http://homepages.uconn.edu/teplyaev/}

\thanks{Research supported in part by NSF grant  DMS 1106982}
\begin{document}

\begin{abstract}
This article studies potential theory and spectral analysis on compact metric spaces, which we refer to as fractal quantum graphs. These spaces can be represented as a (possibly infinite) union of 1-dimensional intervals and a totally disconnected (possibly uncountable) compact set, which roughly speaking represents the set of junction points. 
Classical quantum graphs and fractal spaces such as the Hanoi attractor are included among them. 
We begin with proving the existence of a resistance form on the Hanoi attractor, and go on to establish heat kernel estimates and upper and lower bounds on the eigenvalue counting function of Laplacians corresponding to weakly self-similar measures on the Hanoi attractor.
These estimates and bounds rely heavily on the relation between the length and volume scaling factors of the fractal.
We then state and prove a necessary and sufficient condition for the existence of a resistance form on a general fractal quantum graph. Finally, we extend our spectral results to a large class of weakly self-similar fractal quantum graphs. 
\setcounter{tocdepth}{1}
\tableofcontents
\end{abstract}\date{\today}
\maketitle
\section{Introduction}

The quantum graphs provide important examples in mathematical physics literature, see \cite{BK13} and references therein. 
However the notion of fractal quantum graphs has eluded the researches in mathematics and physics for long time. 
This paper presents new results 
on 
the existence of a well defined energy (resistance/Dirichlet) form 
on fractal quantum graphs. 
In order to show our construction in detail, we begin with a concrete example embedded in $\mathbb R^2$, and also present a general theory.   
In addition, we show how to 
 obtain the spectral asymptotics, i.e. asymptotic estimates of the eigenvalue counting function  and the spectral dimension in particular, for 
a 
 class of weakly self-similar examples. 
 Our interest in fractal quantum graphs spectral problems is 
 motivated by recent physics applications, see \cite{Akkermans, Dunne, DGV}.

 By introducing the concept of fractal quantum graphs we would like to connect physics li\-te\-ra\-ture on fractals (see e.g.~\cite{BCD+08,ADT10,ABD+12}) with quantum graphs. The mo\-dern theory of quantum graphs and its connection to quantum chaos was started in~\cite{KS97} and has been discussed in~\cite{KS02,KS03,Ku04}, see also~\cite{GS06,BK13} for an exhaustive review on this field. \color{black}Fractal networks in particular have been of interest in the study of superconductivity \cite{Ale83, AH83}. Our work in a sense partially generalizes the dendrite fractals considered in \cite{Kig95}, see also \cite{Croydon} and references therein. Note also recent topological results on very similar spaces in~\cite{Geo14} as well as construction of Brownian motion on them in~\cite{GK14}. In our work we attempt to appeal to two different communities that in present have small intersection: the fractal analysis community and the quantum graph community, hopefully generating a bidirectional flow of ideas from both fields.
 
%
%
%
%

Resistance forms, see Definition \ref{def-res}, have been very useful in the study of analysis on fractals from an intrinsic point of view, in particular with Jun Kigami's work on post-critically finite self-similar fractals in~\cite{Kig93} and \cite[Chapter 2]{Kig01}. 
Resistance forms are bilinear forms which induce an effective resistance  between points, analogous to that of electrical networks, and this resistance defines a me\-tric on the space. 
With an appropriate measure, a resistance form is a Dirichlet form on the $L^2$ space associated with that measure. In this way a choice of measure also induces a self-adjoint operator and a symmetric Markov process. The behavior of the eigenvalue counting function  determines the spectral dimension of these fractals while the heat kernel estimates indicate the behavior of the stochastic process. These objects also determine physical aspects of the space \cite{ADT09,bAV00}. 
Heat kernel estimates for 
resistance forms and other related questions are discussed in~\cite{Kig12}. 
In~\cite{HT12,IRT12} a general theory of intrinsic geometric analysis is developed for Dirichlet spaces in general, in~\cite{HKT13} this is applied to resistance forms and length structures, and to  differential equations on these spaces.
%
 Many examples of resistance forms come from finitely ramified, mostly self-similar, cases~\cite{HMT06,FST06,BCF+07,Pei08,MST04,Tep08}. By important technical reasons, these results are not directly applicable to fractal quantum graphs, mainly because it is more natural to approximate these spaces by quantum graphs rather than discrete networks. However, we are able to prove the existence of a resistance form on some large class of  spaces which have no a priori self-similarity, and give a concrete description of its domain.


\begin{figure}[h!tpb]
\includegraphics[width=3.3cm,height=2.8cm]{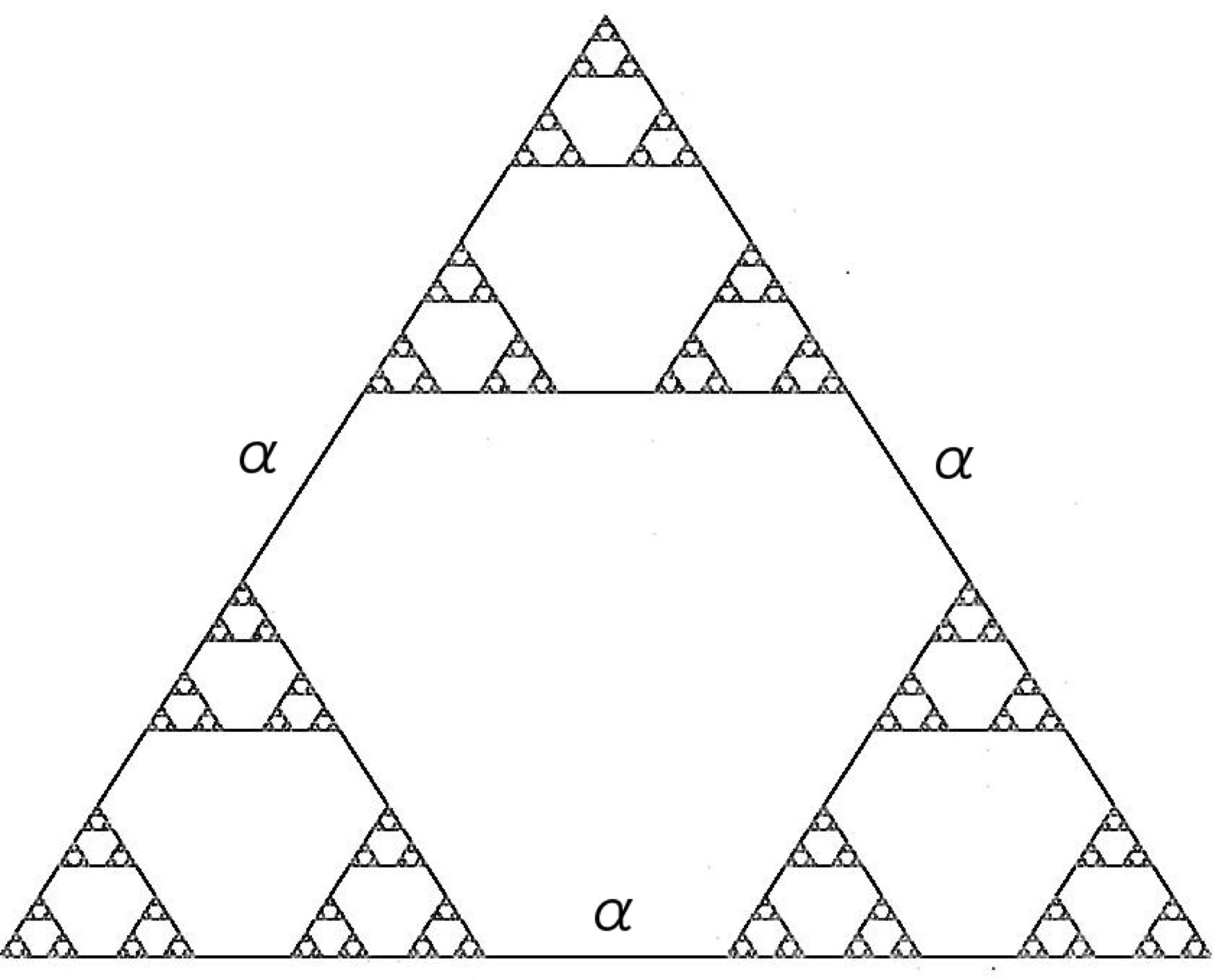}
\caption{The Hanoi attractor.}
\label{HanoiAtt}
\end{figure}

We start by analyzing a particular quantum graph, the \textit{Hanoi attractor of parameter} $\alpha$ that we denote by $\Ka$ (see Figure~\ref{HanoiAtt}). This is a non self-similar fractal, where the parameter $\alpha$ can be understood as a length scaling parameter as it is the length of the three longest segments joining copies of $\Ka$. 
This parameter substantially affects the properties of $\Ka$: when $\alpha=0$, $\Ka$ coincides with the Sierpi\'nski gasket; if $\alpha\in(0,1/3)$, then $\Ka$ has fractional Hausdorff dimension $\ln 3/\ln 2-\ln (1-\alpha)$; if $\alpha\in [1/3, 1)$ we obtain a 1-dimensional object; and if $\alpha=1$, then $\Ka$ is an equilateral triangle. The geometric properties of $\Ka$ were studied in~\cite{ARF12}. 

To study spectral asymptotics, it is necessary to consider spaces which are self-similar in a weak sense, such as the Hanoi attractors and their higher dimensional generalizations. For the Hanoi attractor with parameter $\alpha\in (1/3,1)$, the choice of measure is naturally the 1-dimensional Hausdorff measure, i.e. length measure.  If $\alpha\in (0,1/3)$, having Hausdorff dimension strictly greater than 1 complicates the analysis, although every point has a neighborhood isometric to an interval with the exception of a totally disconnected (i.e. topologically 0 dimensional) set. To deal with these issues we introduce finite self-similar measures on $\Ka$.
The main technique that  we use to obtain the spectral asymptotics  is the 
standard Dirichlet-Neumann bracketing, see~\cite{KL93,Kaj10}. These arguments, informally speaking, use the fact that small-scale metric pro\-per\-ties correspond to larger eigenvalues. Weak self-similarity is therefore critical in achieving spectral asymptotics, as it allows us to infer pro\-per\-ties of the fractal at arbitrarily small scales.

\section{Main Results}

After recalling some basics of the theory of metric and quantum graphs in Section 3, Section 4 is devoted to the approximation of any Hanoi attractor $X:= \Ka$ by metric graphs. 

Let $\ell_1,\ell_2,\ell_3\subset X$ denote the line segments of length $\alpha\in(0,1]$, and let $F_1(X),F_2(X),F_3(X)$ denote the first-level copies of $X$. The scaling length of these copies is
\begin{equation}\label{e-r}
r:=\frac{1-\alpha}{2}
\end{equation} 
 see Figure \ref{FirstLevelCopies}({\sc a}).\begin{figure}[h!tpb]
\begin{subfigure}{.4\linewidth}
\centering
\includegraphics[width=3.2cm,height=2.8cm]{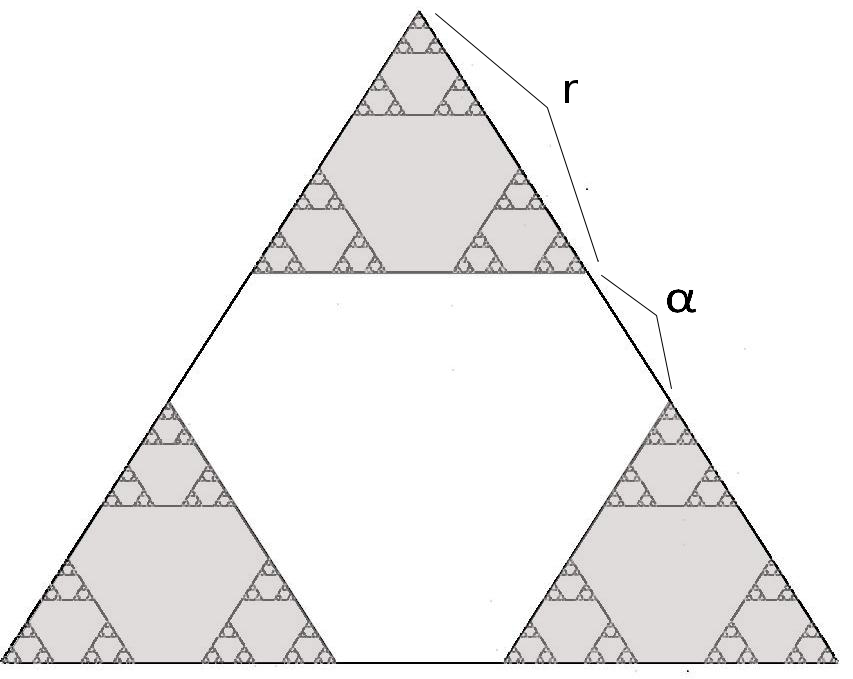}
\caption{Length scaling parameters}
\label{fig:sub1-A}
\end{subfigure}%
\begin{subfigure}{.4\linewidth}
\centering
\includegraphics[width=3.2cm,height=2.8cm]{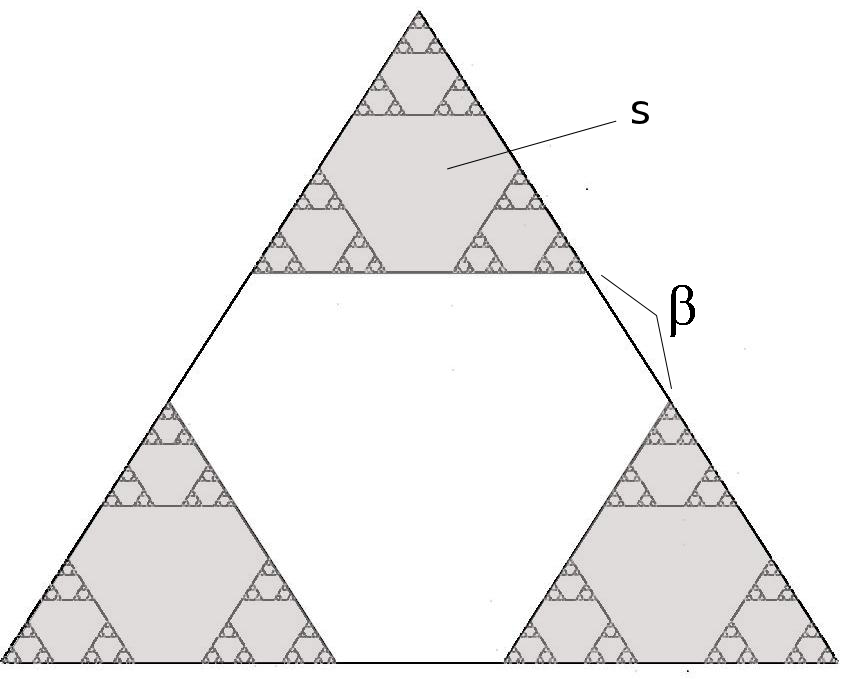}
\caption{Measure scaling parameters}
\end{subfigure}
\caption{Length and measure scaling parameters in the Hanoi attractor.}
\label{FirstLevelCopies}
\end{figure}We are interested in the resistance form, in the sense of Kigami \cite{Kig01}, that satisfies 
the natural scaling relation, see also Lemma~\ref{lem-r},  
\begin{equation}\label{e-Escaling}
\eng(u, u) = \sum_{i = 1}^3 r^{-1}\eng(u{\circ}F_i, u{\circ}F_i) + \sum_{i = 4}^6\int_{F_i(K)} (u')^2dx
\end{equation}
where contractions $F_i$ are defined in Section 4. Here, $dx$ represents the usual one dimensional integral along three straight line segments 
of length $\alpha$
in $X$, and $u'$ represent the usual derivatives along these straight line segments.
Combining the standard theory of electrical networks, including the so co called Delta-Y and Y-Delta transforms in Figure~\ref{deltay-}, with the abstract theory of Kigami's resistance forms   \cite{Kig01}, one can obtain the following standard proposition. 
\begin{prop*}
For any $\alpha\in(0,1]$ and $r$ given by the relation \eqref{e-r} there is a unique resistance form $\eng$ on $\Ka$ that satisfies  \eqref{e-Escaling}.
\end{prop*}In Sections 4, 5, 6 of our paper we obtain much more information about this resistance form. For instance, the effective resistance between the corner points of the triangle is given by \eqref{e-Ra}.

\begin{remark}

We would like to emphasize that the relation \eqref{e-r} is dictated by the Euclidean geometry of $\mathbb R^2$. 
If one considers $\Ka$ as  an abstract topological space, 
then for any $\alpha>0,R>0$ one can build a unique resistance form  that   satisfies a relation  
$$ 
\eng(u, u) = \sum_{i = 1}^3 r^{-1}\eng(u{\circ}F_i, u{\circ}F_i) + R^{-1}\sum_{i = 4}^6\int_{F_i(K)} (u')^2dx
$$ 
if and only if $r\in[0,\frac35)$, which is a larger range than that dictated by \eqref{e-r}. Such questions 
are discussed in \cite{Kigami-note} and 
will be subject of future work.  
In Sections 4 and 5 we focus on the relation of $\eng$ and the Euclidean geometry, 
and in Section 8 we 
explain how to generalize it to a wider class of examples, which do not have to be self-similar in any sense. 
\end{remark}


We 
relate the 
resistance form $\eng$ with   the energy on $X$ that comes from the expression
\begin{equation} \label{e-one}
\eng(u,v):=\int_X u'v'\,dx 
\end{equation} 
%
for continuous functions which are differentiable when restricted to line segments. 
Here $dx$ represents the usual one dimensional integral along the countably many straight line segments in $X$, and $u',v'$ represent the usual derivatives along these straight line segments.

\begin{definition} We say  that $u\in H^1(X)$ if and only if $u\in C(X)$, the restriction of $u$ to any straight line segment is an $H^1$ function on that segment, and 
formula \eqref{e-one} gives $\eng(u,u)<\infty$. \end{definition} 

\begin{theorem}\label{thm resistance form}
$(\eng, H^1(X))$ is the resistance form on $X$ that satisfies \eqref{e-Escaling}.
\end{theorem}

{
The properties of the domain $H^1(X)$ of $\eng$ are very delicate. For instance, if $\alpha\in (1/3, 1)$, then the restriction of any $C^1(\R^2)$ function 
 to $X$ is in $H^1(X)$. However, this is not the case when $\alpha \in(0,1/3)$ because the total length of $X$ is infinite (see Remark~\ref{remarkDef}), and so a generic $u\in C^1(\R^2)$ will have infinite energy.
}

Section~\ref{sec-asymp} deals with the behavior of the eigenvalue counting function of the Laplacian associated to the Dirichlet form induced by the resistance form $(\eng,H^1(X))$ on $L^2(X,\mu)$, where $\mu$ is the 
unique 
weakly self-similar 
re\-gu\-lar probability measure on $X$ defined 
with measure scaling weights $\beta,s$ satisfying
\begin{equation}\label{eq: mu=1}
1=\mu(X)=3\beta+3s,
\end{equation}
where $\beta:=\mu(\ell_i)$ and $s:=\mu(F_i(X))$, $i=1,2,3$; 
$\mu$ is a scaled version of Lebesgue measure on $\ell_i$ and a scaled copy of itself on $F_i(X)$. In this way, $s^n=\mu(F_w(X))$ for any $n$-level copy of $X$, $F_w(X)$.
From~\eqref{eq: mu=1} we obtain that 
\begin{equation*}\label{eq: def of s}
s=\frac{1-3\beta}{3}
\end{equation*}
and hence $\beta\in(0,1/3)$. Note that if $\beta=1/3$, then $\mu(F_i(X))=0$, and thus the support of $\mu$ would not be all of $X$. If $\beta = 0$, then $\mu$ is the restriction of the $-\log 3/\log r$-dimensional Hausdorff measure on $\R^2$ to $X$. In this situation, the measure of any line segment is $0$, which is also undesirable. These assumptions will be briefly recalled at the beginning of Section 6.

Among our main  results are constructive polynomial estimates of the \textit{eigenvalue counting function} of the Laplacian associated to the Dirichlet form induced by $(\eng, H^1(X))$ under Dirichlet --resp. Neumann-- boundary conditions. As boundary of $X$ we consider the set $V_0$ which consists of the three vertices of the equilateral triangle where $X$ is embedded.

\begin{theorem}\label{theorem: d_S X}
Let $rs=\frac{1}{6}(1-\alpha)(1-3\beta)$, where $\alpha$ is the length scaling factor of $X$, and $\beta$ is the volume scaling factor of $\mu$. There exist constants $C_1,C_2>0$ and $x_0>0$ such that
\begin{itemize}
\item[(i)] if $0<rs<\frac{1}{9}$, then
\[C_1x^{\frac{1}{2}}\leq N_D(x)\leq N_N(x)\leq C_1 x^{\frac{1}{2}},\]
\item[(ii)] if $rs=\frac{1}{9}$, then
\[C_1x^{\frac{1}{2}}\log x\leq N_D(x)\leq N_N(x)\leq C_2 x^{\frac{1}{2}}\log x,\]
\item[(iii)] if $\frac{1}{9}<rs<\frac{1}{6}$, then
\[C_1x^{\frac{\log 3}{-\log (rs)}}\leq N_D(x)\leq N_N(x)\leq C_2 x^{\frac{\log 3}{-\log (rs)}}\]
\end{itemize}
for all $x>x_0$.
%
%
%

In particular, 
\begin{equation*}
d_S X=\left\{\begin{array}{rl}
						1,& 0<rs\leq \frac{1}{9},\\ 
						\frac{\log 9}{-\log (rs)},&\frac{1}{9}<rs<\frac{1}{6}.
\end{array}\right.
\end{equation*}
\end{theorem}

This result shows us that both the metric and the measure parameter strongly affect the spectral properties of the operator.

Another way of understanding $X$ is as a graph-directed fractal, introduced in~\cite{MW88} and treated analytically in~\cite{HN03}. Limiting spectral asymptotics, and in particular the spectral dimension, for $X$ can be deduced from~\cite{HN03}. However, the above theorem provides estimates for $N_D$, which is a stronger result. 

The approach here is different than in~\cite{ARF14}, where the resistance form was based on a totally disconnected fractal subset of $X$ (a kind of ``fractal dust'') connected by inserting one dimensional conductances. The main term in the spectrum was that of the ``fractal dust'' and in a sense equivalent to the usual Sierpinski gasket. In our current analysis we do not consider energy supported on any zero-dimensional fractal part but just quantum graph edges, providing anything else with measure and resistance zero.
\color{black} 

Section \ref{sec HKEs} discusses the behavior of the heat kernel with respect to various measures. If $\Delta$ is the generator of a Dirichlet form $\eng$, then the heat kernel is the integral kernel of the heat semi-group. More explicitly, the function $p: \R^+ \times X\times X \to \R$ is the heat kernel of $\Delta$ if the heat equation
\[
\begin{cases}
\Delta u(t,x) = \frac{du}{dt}(t,x) \\ 
u(x,0) = f(x) 
\end{cases}
\]
is solved by $u(t,x) = \int_X p(t,x,y) f(y) \ dy$.

When $1/3<\alpha < 1$, $X$ has finite length and thus $(\eng,H^1(X))$ is a Dirichlet form with respect to the Hausdorff $1$-measure $\haus$. If $p$ is the heat kernel for this Dirichlet form, $p$ satisfies Gaussian estimates
\[
{c_1}{t^{-1/2}} \exp\paren{-\frac{c_2|x-y|^2}{t}} \leq p(t,x,y) \leq {c_3}{t^{-1/2}} \exp\paren{-\frac{c_4|x-y|^2}{t}}
\]
for some $c_1,c_2,c_3,$ and $c_4$. Note that Theorem~\ref{theorem: d_S X}~(i) is applicable here. This is contrasted with the sub-Gaussian estimates associated with many fractal spaces, see \cite{BN98,Kig12}.

Section 8 answers the question about existence of resistance forms in a more general framework. Here, a~\textit{fractal quantum graph} consists of a separable compact connected locally connected metric space $(X,d)$ together with a sequence of lengths $\{\ell_k\}_{k=1}^{\infty}\subset (0,\infty)$ and isometries $\Phi_k\colon[0,\ell_k]~\to~X$ such that
\[
X\setminus\bigcup_{k=1}^{\infty}\Phi_k((0,\ell_k))
\]
is totally disconnected.
Conditions are given that ensure the existence of a resistance form on $X$ that behaves like the 1-dimensional Dirichlet energy on each sub-interval. Quantum graphs, Hanoi attractors and generalized Hanoi-type quantum graphs satisfy these assumptions.


Last section presents the so--called \textit{generalized Hanoi-type quantum graphs} $X_{N_0,\alpha}$. In this case, $N_0$ can be understood as a ``dimension parameter'' because $\dim_H X_{N_0,\alpha}\leq N_0-1$, while $\alpha$ is again the length of the longest segments in $X_{N_0,\alpha}$. This parameter will be chosen to lie in the interval $\big(0,\frac{N_0-2}{N_0}\big)$ so that we deal with a fractal object.

The construction of the resistance form $(\eng,\Dom\eng)$ in this case is carried out in the same way as in Section 5. In order to get a Dirichlet form out of it, we introduce a measure on $X_{N_0,\alpha}$ depending again on a parameter $\beta$ that measures the masses of segments of length $\alpha$.


By analogous arguments as in Section 5, 6 and 7, we obtain the following spectral asymptotics of the Laplacian associated to the Dirichlet form induced by $(\eng,\Dom\eng)$.

\begin{theorem}\label{theorem: d_S X_N_0}
Let $r=\frac{1-\alpha}{2}$ and $s=\frac{2-N_0(N_0-1)\beta}{2N_0}$. There exist constants $C_1,C_2>0$ and $x_0>0$ such that
\begin{itemize}
\item[(i)] if $0<rs<\frac{1}{N_0^2}$, then
\[C_1x^{\frac{1}{2}}\leq N_D(x)\leq N_N(x)\leq C_1 x^{\frac{1}{2}},\]
\item[(ii)] if $rs=\frac{1}{N_0^2}$, then
\[C_1x^{\frac{1}{2}}\log x\leq N_D(x)\leq N_N(x)\leq C_2 x^{\frac{1}{2}}\log x,\]
\item[(iii)] if $\frac{1}{N_0^2}<rs<\frac{1}{2N_0}$, then
\[C_1x^{\frac{\log N_0}{-\log (rs)}}\leq N_D(x)\leq N_N(x)\leq C_2 x^{\frac{\log N_0}{-\log (rs)}}\]
\end{itemize}
for all $x>x_0$.
%
%

In particular, 
\begin{equation*}
d_S X=\left\{\begin{array}{rl}
						1,& 0<rs\leq \frac{1}{N_0^2},\\ 
						\frac{\log N_0^2}{-\log (rs)},&\frac{1}{N_0^2}<rs<\frac{1}{2N_0}.
\end{array}\right.
\end{equation*}
\end{theorem}

\subsection*{Acknowledgments} 
The authors thank Pavel Kurasov 
and David Croydon for  helpful input concerning quantum graphs.
DJK thanks Leonard Wilkins for useful conversations. 


\section{Abstract quantum graph basics}\label{quantumgraphs}

A graph $\graph = (V,\edge, \partial)$ is a finite set of vertices $V$ with a finite set of edges $\edge$ and a map $\partial: \edge \to V\times V$ given by $\partial e := (\partial_- e ,\partial_+e)$. A weighted graph has the additional structure of $\res:\edge\to (0,\infty)$. The weight, or conductance, of an edge $e$ is the quantity $1/r(e)$, thus $r(e)$ is the resistance of the edge $e$.
A metric graph $\gmet$ is the CW 1-complex with set of $0$-cells $V$ and the set of $1$-cells indexed by the edges with endpoints given by $\partial_\pm$. $\gmet$ is covered by the maps
$\Phi_e:I_e\to \gmet$, $I_e = [0,r(e)]$, $e\in E$, such that 
$$
{\Phi_e|_{(0,r(e))}: (0,r(e))
\to \Phi_e\big((0,r(e))\big)
}
$$
is a homeomorphism onto its image, and $\Phi_e(I_e)$ is the $1$-cell associated to the edge $e$.  $\gmet$ is given a metric and a measure $m$ which is induced by $\Phi_e$. 

The space of $L^p$ functions on $\gmet$ is defined by 
\[
L^p(\gmet) := \bigoplus_{e\in\edge} L^p(I_e),
\]
where $L^p(I_e)$ is the classical $L^p$ space on $I_e$ with respect to the Lebesgue measure.  We identify $L^p(\gmet)$ with functions on $\gmet$ by the maps $\Phi_e$ (notice that $V$ is a set of measure $0$).

The Sobolev space on $\gmet$ is defined by
\[
\quad H^n(\gmet) := C(\gmet)\cap\bigoplus_{e\in\edge} H^n(I_e),
\]
where $H^n(I_e)$ is the classical Sobolev space $H^n$ on the interval $I_e$, i.e. $f\in H^n(\gmet)$ if and only if $f\in C(\gmet)$ and $f\circ \Phi_e \in H^n(I_e)$ for all $e\in E$. In particular, $H^1(\gmet)$ is the domain of the Dirichlet energy with standard boundary conditions,
\[
\eng_\gmet(f,g) := \sum_{e\in E} \int_0^{r(e)} \paren{f\circ\Phi_e}'\paren{g\circ\Phi_e}' \ dt.
\]
%
%
A quantum graph is a metric graph with either the above energy form, or the associated self-adjoint (Laplacian) operator on $\gmet$. 

Further details on metric and quantum graphs can be found in the book~\cite{BK13}.
{Note that in our paper we, for the sake of convenience, mostly consider 
quantum graphs embedded in an Euclidean space. However in Section~\ref{sec-frg} we study a more abstract setup. }

\section{Definitions of Hanoi attractors}\label{sect HanoiAtts}
In this section we briefly recall the definition of Hanoi attractors and approximate them by quantum graphs. 


Let $\alpha\in(0,1/3)$ and let $p_1,\ldots p_6\in\R^2$ be the 
fixed points of the mappings
\begin{align*}\label{eq: defHanoiAtt}
F_{\alpha,i}\colon&\R^2\longrightarrow\hspace*{0.3cm}\R^2&\nonumber\\
& x\hspace*{0.2cm}\longmapsto A_i(x-p_i)+p_i,&i=1,\ldots 6,
\end{align*}
where 
\begin{align*}
&A_1=A_2=A_3=\frac{1-\alpha}{2}\begin{pmatrix} 1&0\\0&1\end{pmatrix},&
&A_4=\frac{\alpha}{4}\left( \begin{array}{lr}
                           1&-\sqrt{3}\\
			   -\sqrt{3}&3
                           \end{array}\right),\\
&A_5=\alpha\left( \begin{array}{lr}
                            1&0\\
			    0&0
                            \end{array}\right),&
&A_6=\frac{\alpha}{4}\left( \begin{array}{lr}
                            1&\sqrt{3}\\
			    \sqrt{3}&3
                            \end{array}\right).
\end{align*}
Since the contraction ratios $r_i$ of each $F_{\alpha,i}$ satisfy $r_i\in(0,1)$, $\{F_{\alpha,i}\}_{i=1}^6$ is a family of contractions and for the iterated function system $\{\R^2;F_{\alpha,1},\ldots F_{\alpha,6}\}$ there exists a unique non-empty compact set $\Ka\subset\R^2$ such that
\[\Ka=\bigcup_{i=1}^6 F_{\alpha,i}(\Ka).\]
This set is called the \textit{Hanoi attractor of parameter} $\alpha$. The pa\-ra\-me\-ter $\alpha$ will be arbitrary but fixed, thus to simplify notation we will write $X:=\Ka$ and $F_i:=F_{\alpha,i}$ for each $i=1,\ldots ,6$. $X$ is not strictly self-similar because the contractions $F_4$, $F_5$ and $F_6$ are not similitudes. However, this fractal still has some weak self-similarity due to the similitudes $F_1$, $F_2$ and $F_3$.

Let us denote by $\A$ the alphabet on the symbols $1,2,3$. For each word $w=w_1\cdots w_n\in\An$, $n\in\N$, we define 
\[F_w (x):=F_{w_1}\circ F_{w_2}\circ\cdots \circ F_{w_n}(x),\qquad\,x\in\R^2,\]
and $F_{w_0}:=\operatorname{id}_{\R^2}$ for the empty word $w_0$. For any $w \in \An$, $F_w(X)$ is homeomorphic to $X$.

In a natural sense
, we approximate $X$ by the metric graphs $\Gmet$ determined by $(V_n,E_n,\partial,r)$ and defined below. 

\begin{definition}\label{def ApproxMetricGraphs}
For any $n\in\N_0$, we define the vertex set
\begin{equation*}
V_n:=\bigcup_{w\in\An} F_{w}(\{p_1,p_2,p_3\})
\end{equation*}
and the edge set $E_n:=T_n\cup J_n$, where
\begin{align*}
T_n&:=\{ \{x,y\}\,\vert\,\exists\,w\in\An\text{ s.t. }x,y\in F_w(V_0)\},\\
J_n&:=\{ \{x,y\}\,\vert\,\exists\,0<k<n, w\in\A^{k-1}\text{ s.t. }x=F_{wj}(p_i),y=F_{wi}(p_j),\,i,j\in\A,i\neq j\}.
\end{align*}
Moreover, let $r\colon E_n\to (0,\infty)$ be the weight function given by the edge length, i.e.
\begin{equation*}\label{eq: weightfunction}
r(e):=\left\{\begin{array}{rl}
		\left(\frac{1-\alpha}{2}\right)^n,&\text{for } e\in T_n,\\
		\alpha\left(\frac{1-\alpha}{2}\right)^k,&\text{for }e\in J_k\setminus J_{k-1},~1\leq k\leq n.
\end{array}\right.
\end{equation*}
$G_n := (V_n,E_n,\partial,r)$ is a weighted graph with any orientation $\partial$ and we define the metric graph $\Gmet_n$ associated to $G_n$ as a subset of $X$ where $\Phi_e:I_e \to \R^2$ is given by
\[
\Phi_e(t) = t\partial_+e - (r(e)-t)\partial_-e.
\]
Notice that $\Gmet_n$ is a subset of $\Gmet_{n+1}$ and $V_n\subset V_{n+1}$, however $E_n$ is not a subset of $E_{n+1}$ and thus $G_n$ is not a subgraph of $G_{n+1}$. 
\end{definition}

In the set $E_n$ we distinguish two different types of edges: on one hand, $T_n$ contains ``triangle-type'' edges, i.e. edges building a triangle. On the other hand, $J_n$ denotes the set of ``joining-type'' edges, which join the triangles built by the edges in $T_n$.

We equip these graphs with the measure $m$ introduced in Section~\ref{quantumgraphs}, which coincides with the $1$-dimensional Hausdorff measure. Hence, $(\Gmet_n)_{n\in\N_0}$ is a sequence of metric graphs that approximates $X$ as Figure~\ref{approx} suggests 
{in the sense that}
\[X=\operatorname{cl}\left(\bigcup_{n\in\N_0}\Gmet_n\right)=\operatorname{cl}\left(\bigcup_{n\in\N_0} \bigcup_{e\in E_n}\Phi_e(I_e) \right),\]
where $\operatorname{cl}(\cdot)$ means closure with respect to the Euclidean metric. Later on we will show in Theorem~\ref{lemR} that on $X$, the Euclidean and the effective resistance topology coincide. 

\begin{figure}[h!tpb]
\begin{tabular}{ccccc}
\begin{tikzpicture}[scale=0.35]
\coordinate (p_1) at (0,0);
\coordinate (p_2) at (3,5.1961);
\coordinate (p_3) at (6,0);

\draw (p_1) -- (p_2) -- (p_3) -- cycle;
\end{tikzpicture}
&
\begin{tikzpicture}[scale=0.35]
\coordinate (p_1) at (0,0);
\coordinate (p_2) at (3,5.1961);
\coordinate (p_3) at (6,0);
\coordinate (p_12) at (1,1.732);
\coordinate (p_21) at (2,3.4641);
\coordinate (p_13) at (2,0);
\coordinate (p_31) at (4,0);
\coordinate (p_23) at (4,3.4641);
\coordinate (p_32) at (5,1.732);

\draw (p_1) -- (p_2) -- (p_3) -- cycle;
\draw (p_12) -- (p_13)  (p_32) -- (p_31)   (p_23) -- (p_21);

\end{tikzpicture}
&
\begin{tikzpicture}[scale=0.35]
\coordinate (p_1) at (0,0);
\coordinate (p_2) at (3,5.1961);
\coordinate (p_3) at (6,0);
\coordinate (p_12) at (1,1.732);
\coordinate (p_21) at (2,3.4641);
\coordinate (p_13) at (2,0);
\coordinate (p_31) at (4,0);
\coordinate (p_23) at (4,3.4641);
\coordinate (p_32) at (5,1.732);
\coordinate (p_112) at (1/3,1.732/3);
\coordinate (p_121) at (2/3,3.4641/3);
\coordinate (p_113) at (2/3,0);
\coordinate (p_131) at (4/3,0);
\coordinate (p_123) at (4/3,3.4641/3);
\coordinate (p_132) at (5/3,1.732/3);
\coordinate (p_212) at (1/3+2,1.732/3+3.4641);
\coordinate (p_221) at (2/3+2,3.4641/3+3.4641);
\coordinate (p_213) at (2/3+2,3.4641);
\coordinate (p_231) at (4/3+2,3.4641);
\coordinate (p_223) at (4/3+2,3.4641/3+3.4641);
\coordinate (p_232) at (5/3+2,1.732/3+3.4641);
\coordinate (p_312) at (1/3+4,1.732/3);
\coordinate (p_321) at (2/3+4,3.4641/3);
\coordinate (p_313) at (2/3+4,0);
\coordinate (p_331) at (4/3+4,0);
\coordinate (p_323) at (4/3+4,3.4641/3);
\coordinate (p_332) at (5/3+4,1.732/3);

\draw (p_1) -- (p_2) -- (p_3) -- cycle;
\draw (p_12) -- (p_13)  (p_21) -- (p_23)   (p_31) -- (p_32);
\draw (p_112) -- (p_113)  (p_121) -- (p_123)   (p_131) -- (p_132)
(p_212) -- (p_213)  (p_221) -- (p_223)   (p_231) -- (p_232)
(p_312) -- (p_313)  (p_321) -- (p_323)   (p_331) -- (p_332);
\end{tikzpicture}
&
${\cdots}$
&
\includegraphics[width=2.4cm,height=2cm]{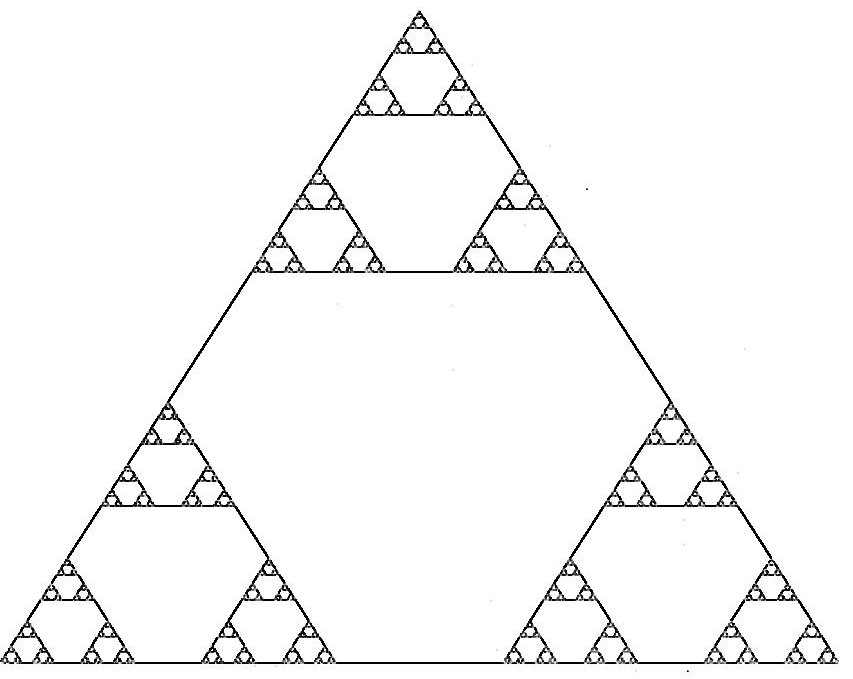}
\end{tabular}
\caption{Metric graphs $\Gmet_0$, $\Gmet_1$, $\Gmet_2$... approximating the Hanoi attractor $X$.}
\label{approx}
\end{figure}


\begin{remark}\label{remarkDef}


The space $(X,m)$ is not finite if $\alpha\in(0,1/3]$ because
\begin{equation*}
m(X)=
{\lim_{n\to\infty}}
\sum_{e\in E_n}m(I_e)
\geq 
{\lim_{n\to\infty}}
\sum_{e\in J_n}m(I_e)
=
3\alpha\sum_{n=1}^{\infty}3^n\left(\frac{1-\alpha}{2}\right)^n=+\infty.
\end{equation*}
Recall that $I_e=[0,r(e)]$ is the interval associated with the edge $e$.

\end{remark}


In order to get a quantum graph out of the metric graph $\Gmet_n$, we consider next a metric graph energy which we denote $\eng_{\Gmet_n}$. 
It is crucial to choose domains $\mathcal{F}_n$, whose functions are everywhere constant except in finitely many ``joining-type'' edges. 
%
Note that 
\[
X\setminus \big(\bigcup\limits_{e\in J_n} \Phi_e(I_e^\circ)\big)=\bigcup\limits_{w\in\A^n}F_w(X),
\]
where $I_e^\circ:=(0,r(e))$ is the interior of $I_e$.
\begin{definition}\label{def LevelEnergy}
We define the domain of functions 
\begin{equation*}
\F_n:=\left\{u\colon X\to\R~\vert~u_{\vert_{\Gmet_n}}\in H^1(\Gmet_n)\text{ and }u_{\vert_{\triangle}}\equiv c_{w}\; \text{for any }\triangle=F_w(X)\right\},
\end{equation*}
where $c_{w}$ are constants that only depend on $\triangle=F_w(X)$, an arbitrary triangular cell indexed by $w\in\A^n$. Note that the non-negative symmetric bilinear form ${\Eng} (u,v)$ given by~\eqref{e-one} is well defined for $
u,v\in \mathcal{F}_*:=\bigcup_{n\in\N_0}\mathcal{F}_n\subsetneqq\Dom \Eng=H^1(X)
$.
We also define the non-negative symmetric bilinear form  
\begin{equation*}\label{e-one-n}
\eng_{\Gmet_n}\colon H^1(\Gmet_n)\times H^1(\Gmet_n)\to\R,
\ \ \ \ \ \qquad 
\eng_{\Gmet_n}(u,v):=\int_{\Gmet_n} u'v'\,dx, 
\end{equation*}
 and call it the standard \textit{energy form on $\Gmet_n$}. 
\end{definition}

\begin{remark}
The formulas for $\Eng$ and $\eng_{\Gmet_n}$ are very similar, but differ in their domains of definition. This will be crucial in the following analysis. We shall use the suggestive notation $\Eng_n$ for the following expressions  
\[
\Eng_n(u,v):=\Eng(u,v)=\eng_{\Gmet_n}(u|_{\Gmet_n},v|_{\Gmet_n}),\qquad u,v\in\mathcal{F}_n,
\]
which are well defined and equal  for all $n\in\N_0$. Again, notice that the only difference between $\eng_n$ and $\eng$ or $\eng_{\Gmet_n}$ is the domain.
\end{remark}

\section{Energy on Hanoi attractor is a resistance form}\label{resistance}
In this section we prove that $(\Eng,H^1(X))$ is a resistance form on $X$ in the sense of~\cite{Kig12}:

\begin{definition}\label{def-res}
Let $X$ be a set. A pair $(\mathcal{E},\Dom\mathcal{E})$ is called a resistance form if 
\begin{enumerate}[(RF 1)]
\item $\mathcal{E}$ is a non-negative symmetric bilinear form on $\Dom \mathcal{E}$, a linear subspace of $\ell(X):=\{u\colon X\to \R\}$ that contains constants, and $\mathcal{E}(u,u) = 0$ if and only if $u$ is constant on $X$.
\item If $\sim$ is the equivalence relation in $\Dom \mathcal{E}$ where $u \sim v$ iff $u-v$ is constant, then $(\Dom\mathcal{E}/_{\sim},\mathcal{E})$ is a Hilbert space.
\item\label{iRF3} For any two points $p\neq q$ in $X$, there exists $u\in\Dom \mathcal{E}$ such that $u(p)\neq u(q)$.
\item For any $p,q\in X$,
\[
\sup\set{\frac{(u(p)-u(q))^2}{\mathcal{E}(u,u)}~\colon~u\in\Dom\mathcal{E}, \mathcal{E}(u,u)\neq0}<\infty
\]
We denote this supremum by $R_{\mathcal{E}}(p,q)$ and call it the effective resistance between $p$ and $q$.
\item (Markov property) For any $u\in\Dom\mathcal{E}$, $\overline{u}\in \Dom\mathcal{E}$ and $\mathcal{E}(\overline{u},\overline{u})\leq\mathcal{E}(u,u)$, where
\[
\overline{u}(p) := \begin{cases}
0 & \text{if }u(p)\leq 0,\\
u(p) & \text{if }0<u(p)<1,\\
1 & \text{if }u(p) \geq 1.
\end{cases}
\]
\end{enumerate}
\end{definition}

Note that $\eng_{\Gmet_n}$ is  a resistance form on $\Gmet_n$. We would also like to point out that if the condition~(RF \ref{iRF3}) is not satisfied, then $R_{\mathcal{E}}(p,q)$ may equal $0$ even if $p\neq q$. In such a situation the effective resistance is defined but it is not necessarily a metric, yet it can be a pseudometric. This fact is important because when restricted to the functions $u\in\F_n$,  $\Eng_n$  satisfies all the conditions except (RF3). We can define effective resistances $R_n(p,q)$ with respect to $\eng_n$ using the same definition from (RF4) despite the fact that they are not metrics because they are not positive definite. Ne\-ver\-theless they do still satisfy the triangle inequality and hence build a nondecreasing sequence of pseudometrics on $X$. In a certain sense, $\eng_n$ is equivalent to a resistance form on a quotient space of $\Gmet_n$ or of $X$, by identifying all points in a cell $F_w(X)$ in a single point.


\begin{definition}
Let $(\mathcal{E},\Dom\mathcal{E})$ be a resistance form on $X$ and let $S$ be a finite subset of $X$. The resistance form $\trace_S\mathcal{E}\colon\ell(S)\times\ell(S)\to\R$ is given by
\[
\trace_S\mathcal{E}(u,u) := \inf\set{\mathcal{E}(v,v):~v\in\Dom\mathcal{E}, v_{|_S}= u}.
\]
For any $u,v\in\ell(S)$, $\trace_S\mathcal{E}(u,v)$ is defined by applying the polarization identity.
\end{definition}

\subsection{Metric observations}\label{subsec metric obs}
This section establishes the metric properties of $\eng_n$, $\eng_{\Gmet_n}$ and $\eng$, starting with the following simple but important technical observation concerning the resistance form $\eng_{\Gmet_n}$ on ${\Gmet_n}$.  

\begin{lemma}\label{engestimate}
For any points $p,q\in \Gmet_n$ and for any function $u\in H^1(\Gmet_n)$,
\[
|u(p)-u(q)|^2  \leq d_n(p,q) \eng_{\Gmet_n}(u,u),
\] 
where $d_n$ is the intrinsic geodesic distance in $\Gmet_n$ and $\eng_{\Gmet_n}$ is defined in~\eqref{def LevelEnergy}.

Furthermore, for all $v\in \F_n$ and $p,q\in \Gmet$,
\[
|v(p)-v(q)|^2 \leq d_n(p,q) \eng_n(v,v).
\]
\end{lemma}

\begin{proof} The second inequality follows from the first and the fact and that fact that $\eng_n(v,v) = \eng_{\Gmet_n}(v|_{\Gmet_n}, v|_{\Gmet_n})$.

If $p,q$ are both on the same edge, which is a one dimensional straight line segment in $\Gmet_n$, then
\[
|u(p)-u(q)|^2 = \abs{\int_p^q u'(x)  \ dx}^2 \leq \abs{\int_p^q |u'(x)|^2 \ dx}|p-q| \leq \Eng_{\Gmet_n}(u,u)|p-q|.
\]
Here, again, $dx$ represents the usual one dimensional integral along the straight line segments in $X$ and $u'$ and $v'$ represent the usual derivatives along these straight line segments. 
If $p$ and $q$ are not on the same edge, then there are $x_0,\ldots,x_m\in \Gmet_n$ such that $p=x_0$, $q=x_m$, and $x_i$ and $x_{i+1}$ belong to the same edge (these are the vertices which a path from $p$ to $q$ would pass through). 
Then it is easy to see that  
\[
|u(p)-u(q)|^2 \leq \Eng_{\Gmet_n}(u,u)\paren{\sum_{i=0}^{m-1}{|x_{i}-x_{i+1}|} }
\]
by the Cauchy--Schwarz inequality. If we assume that $x_i$ are the vertices traversed by the length minimizing path from $p$ to $q$, then we get the inequality in the lemma. 
\end{proof}

\begin{theorem}\label{lemR}
\begin{enumerate}
\item \label{lemR1}
For any $n\in\N$ and any $p,q\in X$ it holds that $R_{n+1}(p,q)\geq R_n(p,q)$. Moreover, we have the nondecreasing limit 
\begin{equation*}\label{eRs}
0<R(p,q):=\lim_{n\to\infty}R_n(p,q) =\sup_{n}R_n(p,q)<\infty
\end{equation*}
for any distinct $p,q\in X$.  Thus $R$ is a metric on $X$.

\medskip

\item For for any $n\in\N$ and $p,q\in\Gmet_n$, $R_{\Gmet_{n+1}}(p,q)\leq R_{\Gmet_n}(p,q)$. Here we formally define $R_{\Gmet_n}$ to be infinite for points not in $\Gmet_n$. Furthermore, we have a nonincreasing limit 
\begin{equation*}\label{eRi}
0<R(p,q)= \lim_{n\to\infty}R_{\Gmet_n}(p,q) =\inf_{n}R_{\Gmet_n}(p,q)<\infty
\end{equation*}
for any distinct $p,q\in \bigcup_n \Gmet_n$. In particular $(\Gmet_n,R_{\Gmet_n})$ converges to $(X,R)$ in the Gromov--Hausdorff sense.

\medskip

\item\label{i3} There exists a constant $c>1$ such that
\[
\frac{1}{c}\abs{p-q}\leq R(p,q)\leq c\abs{p-q}
\]
for any $p,q\in X$.
\end{enumerate}
\end{theorem}

\begin{proof}
(1) Since $\F_n\subseteq\F_{n+1}$ and $\Eng_n(u,u)=\Eng_{n+1}(u,u)$ for all $u\in\F_n$ we have that
\begin{align*}
R_n(p,q)&
=\sup\Big\{\frac{\abs{u(p)-u(q)}^2}{\Eng_{n+1}(u,u)}:~u\in\F_n, \Eng_n(u,u)\neq0\Big\}\\
&\leq\sup\Big\{\frac{\abs{u(p)-u(q)}^2}{\Eng_{n+1}(u,u)}:~u\in\F_{n+1}, \Eng_{n+1}(u,u)\neq0\Big\}=R_{n+1}(p,q).
\end{align*}
The fact that $0<R(p,q)$ follows from Lemma~\ref{engestimate} and the fact that $\cup_{n=1}^\infty \F_n$ separates points of $X$. $R(p,q)<\infty$ because 
\[
R_n(p,q) = R_n(p',q')=\sup\Big\{\frac{\abs{u(p')-u(q')}^2}{\eng_{\Gmet_n}(u\vert_{\Gmet_n},u\vert_{\Gmet_n})}:~u\in\F_n, \Eng_n(u,u)\neq0\Big\}\leq R_{\Gmet_n}(p',q'),
\]
where $p'$ is chosen to be $p$ is $p\in \Gmet_n$ or any point in $F_w(X)\cap \Gmet_n$ for $w$ being the word of length $n$ such that $p\in F_w(X)$ and $q'$ defined in a similar manner. 
\medskip

(2) Recall  that $R_{\Gmet_n}(p,q):=\sup\big\{\frac{\abs{u(p)-u(q)}^2}{\eng_{\Gmet_n}(u,u)}:~u\in H^1(\Gmet_n), \eng_{\Gmet_n}(u,u)\neq0\big\}$, where
\[
\eng_{\Gmet_n}(u,u)=\sum_{e\in E_n}\int_0^{r(e)}(u\circ\Phi_e)'(x) dx.
\]
Given any function $u\in H^1(\Gmet_{n+1})$, $u\vert_{\Gmet_n}\in H^1(\Gmet_n)$ and $\Eng_{\Gmet_{n+1}}(u,u)\geq\Eng_{\Gmet_n}(u\vert_{\Gmet_n},u\vert_{\Gmet_n})$. Hence
\begin{equation*}\label{eq lemma4 (2)}
\frac{\abs{u(p)-u(q)}^2}{\eng_{\Gmet_{n+1}}(u,u)}\leq\frac{\abs{u(p)-u(q)}^2}{\eng_{\Gmet_n}(u\vert_{\Gmet_n},u\vert_{\Gmet_n})}\qquad\forall\,u\in H^1(\Gmet_{n+1}).
\end{equation*}
Moreover, since any function in $H^1(\Gmet_n)$ can be extended to a function in $H^1(\Gmet_{n+1})$ by in\-ter\-po\-la\-ting on new ``interior'' edges, any function in $H^1(\Gmet_n)$ can be obtained as a restriction of a function in $H^1(\Gmet_{n+1})$ and thus
\begin{multline*}
R_{\Gmet_{n+1}}(p,q)=\sup\Big\{\frac{\abs{u(p)-u(q)}^2}{\eng_{\Gmet_{n+1}}(u,u)}~\colon~u\in H^1(\Gmet_{n+1}), \eng_{\Gmet_{n+1}}(u,u)\neq0\Big\} \\
\leq\sup\Big\{\frac{\abs{u(p)-u(q)}^2}{\eng_{\Gmet_n}(u\vert_{\Gmet_n},u\vert_{\Gmet_n})}~\colon~u\in H^1(\Gmet_n), \eng_{\Gmet_n}(u,u)\neq0\Big\}=R_{\Gmet_n}(p,q).
\end{multline*}
for any $p,q\in X$, and the limit exists.

It remains to be proved that in fact $R(p,q)=\lim\limits_{n\to\infty}R_{\Gmet_n}(p,q)$ for any $p,q\in\bigcup_n \Gmet_n$. 
%
%
If $\mathcal R_n$ is the resistance of a wire in a triangle network such that the restistance between the corners is either $R_n(p_i,p_j)$ or $R_{\Gmet_n}(p_i,p_j)$, with $i,j\in\{1,2,3\}$, $i\neq j$, and $p_1,p_2,p_3$ being the corners of $X$. In either case, the sequence $\set{R_n}_{n=1}^\infty$ satisfies the recurrence relation $\frac53 r\mathcal R_n + \alpha = \mathcal R_{n+1}$. 
\begin{figure}[tbh]
\begin{subfigure}{.25\linewidth}
\centering
\begin{tikzpicture}[scale=2]
\draw (90:1)--(210:1)-- node[below] {$\alpha$}(330:1)--cycle;
\draw ($(90:1/3)+(90:2/3)$)--($(210:1/3)+(90:2/3)$)-- node[below] {$r\mathcal R_n$}($(330:1/3)+(90:2/3)$)--cycle;
\draw ($(90:1/3)+(210:2/3)$)--($(210:1/3)+(210:2/3)$)--($(330:1/3)+(210:2/3)$)--cycle;
\draw ($(90:1/3)+(330:2/3)$)--($(210:1/3)+(330:2/3)$)--($(330:1/3)+(330:2/3)$)--cycle;
\end{tikzpicture}
\caption{}
\label{fig:sub1}
\end{subfigure}%
\begin{subfigure}{.5\linewidth}
\centering
\begin{tikzpicture}[scale=2]
\draw ($(210:1/3)+(90:2/3)$)--($(90:1/3)+(210:2/3)$) ;
\draw ($(330:1/3)+(210:2/3)$) -- node[below] {$\alpha$}($(210:1/3)+(330:2/3)$);
\draw ($(90:1/3)+(330:2/3)$)-- ($(330:1/3)+(90:2/3)$);

\draw ($(90:1/3)+(90:2/3)$)-- node[right] {$\frac{r\mathcal R_n}{3}$} (90:2/3)
	  ($(210:1/3)+(90:2/3)$)-- (90:2/3)
	  ($(330:1/3)+(90:2/3)$)--(90:2/3);
\draw ($(90:1/3)+(210:2/3)$)--(210:2/3)
      ($(210:1/3)+(210:2/3)$)--(210:2/3) 
      ($(330:1/3)+(210:2/3)$)--(210:2/3);
\draw ($(90:1/3)+(330:2/3)$)--(330:2/3)
	  ($(210:1/3)+(330:2/3)$)--(330:2/3)
	  ($(330:1/3)+(330:2/3)$)--(330:2/3);
\end{tikzpicture}
\caption{}
\label{fig:sub2}
\end{subfigure}
\begin{subfigure}{.25\linewidth}
\centering
\begin{tikzpicture}
\draw (90:1)--(210:1)-- node[below] {$\alpha+\frac{2r\mathcal R_n}{3}$}(330:1)--cycle
	  (90:1)-- node[right] {$\frac{r\mathcal R_n}{3}$} (90:2)
	  (210:1)--(210:2)
	  (330:1)--(330:2);
\end{tikzpicture}
\caption{}
\end{subfigure}
\begin{subfigure}{.5\linewidth}
\centering
\begin{tikzpicture}
\draw	  (0,0)-- node[right] {$\frac{5r\mathcal R_n}{9}+\frac{\alpha}{3}$} (90:2)
	  (0,0)--(210:2)
	  (0,0)--(330:2);
\end{tikzpicture}
\caption{}
\end{subfigure}\\[1ex]
\caption{Reduction of the first level approximation network of the Hanoi attractor.}
\label{deltay-}
\end{figure}
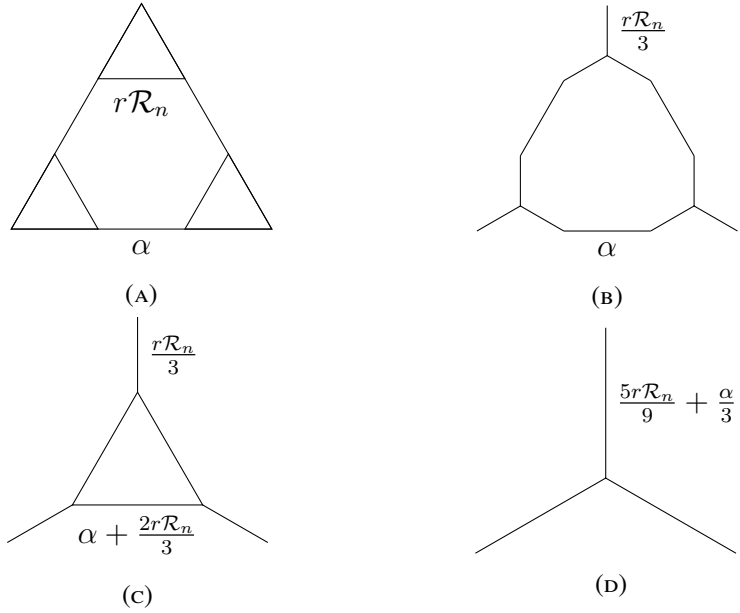
 This can be seen by means of the Delta-Y transform as illustrated in Figure~\ref{deltay-}. Although the Delta-Y transform is classical, one can find the background related to fractal networks in~\cite{BCF+07,IKM+15,MST04,Str06,Tep08}.
The limit must be $2/3$ times the fixed point of the function  $f(z) = \dfrac{5rz}{3}+\alpha $ which is $z_0 = 6\alpha/(1+5\alpha)$ and is thus independent of our choice of sequence. Therefore, the limits coincide and \begin{equation}\label{e-Ra}
R(p_i,p_j) = 2z_0/3 = \dfrac{4\alpha}{1+5\alpha}
\end{equation} 

Applying Kirchhoff's laws, one can now compute the effective resistance between any two points in $\bigcup_{e\in J_n}\Phi_e(I_e)$ for any $n$, and so the limits must coincide for these points as well. Since $\bigcup_{e\in J_n} \Phi_e(I_e)$ becomes uniformly dense in $(\Gmet_n,R_{\Gmet_n})$, this that these metric spaces converge to $(X,R)$ in the Gromov--Hausdorff sense, see for example~\cite[Proposition 7.4.12]{BBI01}.

(3) First, assume that $p,q\in\Phi_e(I_e)$, where $e\in J_n$ and $n$ is the smallest such integer. Then, $e$ is adjacent to $F_w(X)$ for some $w\in\A^n$ and we may assume without loss of generality that $q$ is closer to $F_w(X)$ than $p$ is. Because $n$ is the smallest such that $e\in J_n$, $e$ is the shortest such edge.

This allows us to construct a function $u$ with $u(p) = 0$, $u(q)=1$, interpolating linearly between $p$ and $q$ and staying constant outside. Moreover, $u|_{F_w(X)}\equiv 1$, and it linearly decays from $1$ to $0$ on the other (at most two) edges adjacent to $F_w(X)$. Finally, set $u$ to be constant zero everywhere else. Then, $\Eng_n(u,u)\leq 3/|p-q|$ for all $n\geq n_0$, which implies that $R_n(x,y)\geq |p-q|/3$ and thus $\frac{1}{3}\abs{p-q}\leq R_n(p,q)\leq R(p,q)\leq \abs{p-q}$, where the upper bound comes from Lemma~\ref{engestimate}.

Now suppose that $p,q$ do not belong to such an edge for any $n\in\N$ and let $n_0\in\N$ be the smallest integer such that $p$ and $q$ belong to different $n_0-$cells. Formally, there exist $w_1,w_2\in\A^{n_0}$ with $p\in F_{w_1}(X)$, $q\in F_{w_2}(X)$ and $F_{w_1}(X)\cap F_{w_2}(X)=\emptyset$.
\begin{figure}[h!tpb]
\begin{tikzpicture}[scale=0.55]
\coordinate (p_1) at (0,0);
\coordinate (p_2) at (3.5,6.0621);
\coordinate (p_3) at (7,0);
\coordinate [label=left: {$q'$}](p_12) at (1.5,2.59805);
\fill (p_12) circle (2.5pt);
\coordinate [label=left: {$p'$}](p_21) at (2,3.46405);
\fill (p_21) circle (2.5pt);
\coordinate (p_13) at (3,0);
\coordinate (p_31) at (4,0);
\coordinate (p_23) at (5,3.46405);
\coordinate (p_4) at (5.25,3);
\coordinate (p_32) at (5.5,2.59805);
\draw[draw, fill=black] (2.2,.5) circle (2pt) node[left]{\small{${q}$}};
\draw[draw, fill=black] (3.7,5) circle (2pt) node[left] {\small{${p}$}};
\draw[magenta] (p_1) -- (p_12) -- (p_13) -- cycle;
\draw[blue] (p_2) -- (p_21) -- (p_23) -- cycle;
\draw (p_3) -- (p_32) -- (p_31) -- cycle;
\draw (p_21) -- (p_12)   (p_23) -- (p_32) (p_31) -- (p_13);
\end{tikzpicture}
\caption{{\color{blue}$F_{w_1}(X)$} and {\color{magenta}$F_{w_2}(X)$}}\label{cells}
\end{figure}
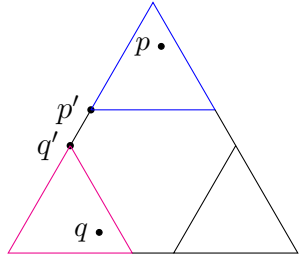
Take $p' \in F_{w_1}(X)$ and $q'\in F_{w_2}(X)$ to be the endpoints of the line segment connecting $F_{w_1}(X)$ and $F_{w_2}(X)$ (see Figure~\ref{cells}). Such points exist because $F_{w_1}(X)$ and $F_{w_2}(X)$ are the largest cells for which $p,q$ are in different cells. Then, $|p-q| \geq |p'-q'| = \alpha\paren{\frac{1-2\alpha}{2}}^{n_0}$ because $p'$ and $q'$ attain the minimum of the (Euclidean) distance between elements in $F_{w_1}(X)$ and $F_{w_2}(X)$. By the triangular inequality
\[
R(p,q)\leq R(p,p')+R(p',q')+R(q',q).
\]
Since $p',q'$ belong to an edge $e\in J_{n_0}$, it follows from Lemma~\ref{engestimate} that $R(p',q')\leq\abs{p'-q'}$. Applying Delta-Y transform we have that
\[
R(p,p')\leq \diam_RF_{w_1}(X)\leq \paren{\frac{1-\alpha}{2}}^{n_0}\diam_R X = \frac23\frac{4\alpha}{1+5\alpha}\paren{\frac{1-\alpha}{2}}^{n_0}=\frac2{3\alpha}\frac{4\alpha}{1+5\alpha}\abs{p-q},
\]
where the second inequality holds because otherwise $p$ and $q$ would have belonged to the same $n_0-$cell. The same holds for $R(q,q')$. 
Since $\frac{4\alpha}{1+5\alpha}<1$, we deduce that $R(p,q)\leq (1 + \frac1\alpha+ \frac1\alpha)\abs{p-q} \leq \frac3\alpha|p-q|$.

On the other hand, $\abs{p-q}<\left(\frac{1-\alpha}{2}\right)^{n_0-1}$ because otherwise $p$ and $q$ could have been separated by $(n_0-1)-$cells and $n_0$ was chosen to be minimal with this property. Note additionally that $R_{n_0}(p,p') = R_{n_0}(q,q') = 0$. Using the bounds from (1) and the lower bound for points which share an edge from above we get
\[
R(p,q)\geq R_{n_0}(p,q) = R_{n_0}(p',q')>\frac{1}{3}\abs{p'-q'}=\frac{1}{3}\alpha\left(\frac{1-\alpha}{2}\right)^{n_0-1}>\frac{\alpha}{3}\abs{p-q}.
\]
Choosing $c=\frac{3}{\alpha}>1$ the chain of inequalities is proved.
\end{proof}


\begin{remark}\label{Eu bilip R}
Theorem~\ref{lemR}~\eqref{i3} proves that $R$ and the Euclidean distance are bi-Lipschitz equivalent, and this implies that the induced topologies on $X$ are the same. In addition, we may define for any $p,q\in X$ the geodesic distance 
\[
d_G(p,q):=\lim_{n\to\infty}d_n(p,q)=\inf_{n}d_n(p,q),
\]
with $d_n$ as in Lemma~\ref{engestimate}. Note that $$d_n(p,q)\geq d_{n+1}(p,q)\geq d_G(p,q)\geq|p-q|$$ for all $p,q\in X$, considering $d_n$ to be infinite if $p,q$ are not in $\Gmet_n$. Moreover, one can prove purely geometrically the sharp bi-Lipschitz estimates
\[
d_G(p,q)\geq|p-q|\geq\frac{1}{2}d_G(p,q),
\]
which also imply that $d_G(p,q)$ is bi-Lipschitz equivalent to $R(p,q)$: Suppose that $n_0\geq 0$ is such that $p,q\in\Gmet_{n_0}$. If $p,q$ belong to the same equilateral triangle, we know from plain geometry that $2|p-q|\geq d_n(p,q)$ for all $n\geq n_0$. If $p,q$ belong to the same hexagon with angles $2\pi/3$ we obtain from this property of equilateral triangles that $2|p-q|\geq d_n(p,q)$ for all $n\geq n_0$ (see Figure~\ref{hexagons}).
\begin{figure}[h!tpb]
\begin{subfigure}{.3\linewidth}
\centering
\begin{tikzpicture}[scale=0.35]
\coordinate (p_1) at (0,0);
\coordinate (p_2) at (3.5,6.0621);
\coordinate (p_3) at (7,0);
\coordinate (p_12) at (1.5,2.59805);
\coordinate [label=below: {$p$}](p_5) at (3.25,0);
\fill (p_5) circle (2.5pt);
\coordinate (p_21) at (2,3.46405);
\coordinate (p_13) at (3,0);
\coordinate (p_31) at (4,0);
\coordinate (p_23) at (5,3.46405);
\coordinate  [label=right: {$q$}](p_4) at (5.25,3);
\fill (p_4) circle (2.5pt);
\coordinate (p_32) at (5.5,2.59805);
\draw[dotted] (p_4) -- (p_5);
\draw[very thin] (p_1) -- (p_2) -- (p_3) --cycle;
\draw[thick] (p_12) -- (p_21) -- (p_23) -- (p_32) -- (p_31) -- (p_13) -- (p_12);
\end{tikzpicture}
\caption{$|p-q|\geq\frac{1}{2}d_{n-1}(p,q)$}
\end{subfigure}
\begin{subfigure}{.3\linewidth}
\centering
\begin{tikzpicture}[scale=0.35]
\coordinate (p_1) at (0,0);
\coordinate (p_2) at (3.5,6.0621);
\coordinate (p_3) at (7,0);
\coordinate (p_12) at (1.5,2.59805);
\coordinate (p_21) at (2,3.46405);
\coordinate (p_13) at (3,0);
\coordinate (p_31) at (4,0);
\coordinate (p_23) at (5,3.46405);
\coordinate (p_32) at (5.5,2.59805);
\coordinate [label=right: {$q$}](q) at (5,1.75);
\fill (q) circle (2.5pt);
\coordinate [label=below: {$p$}](p) at (3.3,0);
\fill (p) circle (2.5pt);
\coordinate [label=below: {$p'$}](p_) at (5.3,0);
\fill (p_) circle (2.5pt);
\draw[densely dashed] (q) -- (p);
\draw[dotted] (q) -- (p_);
\draw[very thin] (p_1) -- (p_2) -- (p_3) --cycle;
\draw[thick] (p_12) -- (p_21) -- (p_23) -- (p_32) -- (p_31) -- (p_13) -- (p_12);
\end{tikzpicture}
\caption{$|p-q|=|p'-q|\geq\frac{1}{2}d_{n-1}(p,q)$}
\end{subfigure}
\begin{subfigure}{.65\linewidth}
\centering
\begin{tikzpicture}[scale=0.35]
\coordinate (p_1) at (0,0);
\coordinate (p_2) at (3.5,6.0621);
\coordinate (p_3) at (7,0);
\coordinate (p_12) at (1.5,2.59805);
\coordinate [label=below: {$p\;$}](p_5) at (3.5,0);
\fill (p_5) circle (2.5pt);
\coordinate (p_21) at (2,3.46405);
\coordinate (p_13) at (3,0);
\coordinate [label=below: {$\;a$}](p_31) at (4,0);
\fill (p_31) circle (2.5pt);
\coordinate [label=right: {$c$}](p_23) at (5,3.46405);
\fill (p_23) circle (2.5pt);
\coordinate [label=above: {$q$}](p_4) at (3.5,3.46405);
\fill (p_4) circle (2.5pt);
\coordinate [label=right: {$b$}](p_32) at (5.5,2.59805);
\coordinate [label=left: {$r$}](r) at (3.5,0.8);
\fill (r) circle (2.5pt);
\draw[densely dashed] (p_4) -- (p_5);
\draw[very thin] (p_1) -- (p_2) -- (p_3) --cycle;
\draw[thick] (p_12) -- (p_21) -- (p_23) -- (p_32) -- (p_31) -- (p_13) -- (p_12);
\draw[dotted] (p_12) -- (p_32) (p_21) -- (p_31) (p_13) -- (p_23);
\end{tikzpicture}
\caption{$|p-q|=|p-r|+|r-q|\geq\frac{1}{2}(|p-a|+|a-r|+|r-c|+|c-q|)=\frac{1}{2}d_n(p,q)$}
\end{subfigure}
\caption{Possible configurations of points in a convex hexagon}
\label{hexagons}
\end{figure}
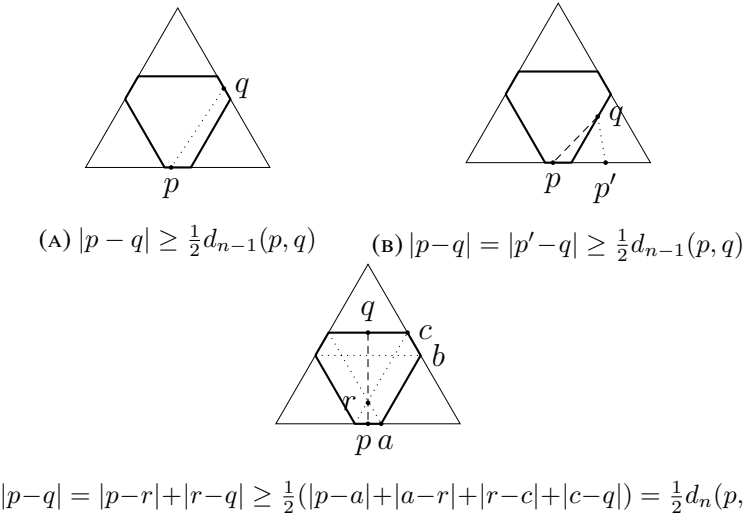
If $p,q\in\Gmet_{n_0}$ are not as in the previous cases, consider the straight line segment connecting $p$ and $q$. It crosses convex sets that are hexagons with angles $2\pi/3$ or equilateral triangles at points $x_1,\ldots,x_m$. If the segment connecting $x_i$ and $x_{i+1}$ lies inside an hexagon, replace it by the piecewise-geodesic going around it (see Figure~\ref{geodesic}). By this procedure we obtain a path inside $\Gmet_{n_0}$ whose length is at most twice $|p-q|$ in view of the previous step. 
Thus we have $2|p-q|\geq d_n(p,q)$ for all $n\geq n_0$ and, conversely, $|p-q|\leq d_G(p,q)$ by definition of geodesic distance.
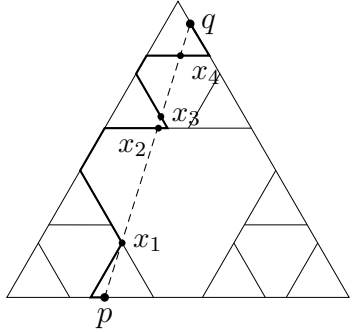
\begin{figure}[h!tpb]
\begin{tikzpicture}[scale=0.65]
\coordinate (p_1) at (0,0);
\coordinate (p_2) at (3.5,6.0621);
\coordinate (p_3) at (7,0);
\coordinate (p_12) at (1.5,2.59805);
\coordinate (p_21) at (2,3.46405);
\coordinate (p_13) at (3,0);
\coordinate (p_31) at (4,0);
\coordinate (p_23) at (5,3.46405);
\coordinate (p_32) at (5.5,2.59805);
\coordinate (p_112) at (1.5*3/7,2.59805*3/7);
\coordinate (p_121) at (2*3/7,3.46405*3/7);
\coordinate (p_113) at (9/7,0);
\coordinate (p_131) at (12/7,0);
\coordinate (p_123) at (15/7,3.4605*3/7);
\coordinate (p_132) at (5.5*3/7,2.59805*3/7);
\coordinate (p_212) at (2+1.5*3/7,3.46405+2.59805*3/7);
\coordinate (p_221) at (2+2*3/7,3.46405+3.46405*3/7);
\coordinate (p_213) at (2+3*3/7,3.46405);
\coordinate (p_231) at (2+4*3/7,3.46405);
\coordinate (p_223) at (2+5*3/7,3.46405+3.46405*3/7);
\coordinate (p_232) at (2+5.5*3/7,3.46405+2.59805*3/7);
\coordinate (p_312) at (4+1.5*3/7,2.59805*3/7);
\coordinate (p_321) at (4+2*3/7,3.46405*3/7);
\coordinate (p_313) at (4+3*3/7,0);
\coordinate (p_331) at (4+4*3/7,0);
\coordinate (p_323) at (4+5*3/7,3.46405*3/7);
\coordinate (p_332) at (4+5.5*3/7,2.59805*3/7);

\coordinate [label=below:{$p$}] (p) at (2,0);
\fill (p) circle (2.5pt);
\coordinate [label=right:{\small{$x_1$}}] (x_1) at (5.5*3/7,2.59805*3/7);
\fill (x_1) circle (2pt);
\coordinate [label=below left:{\small{$x_2$}}] (x_2) at (3.1,3.46405);
\fill (x_2) circle (2pt);
\coordinate [label=right:{\small{$x_3$}}] (x_3) at (3.15,3.7);
\fill (x_3) circle (2pt);
\coordinate [label=below right:{\small{$x_4$}}] (x_4) at (3.55,3.46405+3.46405*3/7);
\fill (x_4) circle (2pt);
\coordinate [label=right:{$q$}] (q) at (3.75,5.6);
\fill (q) circle (2.5pt);
\draw[very thin] (p_1) -- (p_2) -- (p_3) -- cycle;
\draw[very thin] (p_12) -- (p_13)  (p_21) -- (p_23)   (p_31) -- (p_32);
\draw[very thin] (p_112) -- (p_113)  (p_121) -- (p_123)   (p_131) -- (p_132)
(p_212) -- (p_213)  (p_221) -- (p_223)   (p_231) -- (p_232)
(p_312) -- (p_313)  (p_321) -- (p_323)   (p_331) -- (p_332);
\draw[densely dashed] (p) -- (q);
\draw[thick] (p) -- (p_131) -- (x_1) -- (p_12) -- (p_21) -- (x_2) -- (p_213) -- (x_3) -- (p_212) -- (p_221) -- (x_4) -- (p_223) -- (q);
\end{tikzpicture}
\caption{Piecewise-geodesic and Euclidean paths from $p$ to $q$.}\label{geodesic}
\end{figure}
\end{remark}

\begin{remark}
Let $\Omega_{d_G}$ be defined as the completion of $ \bigcup_{n\in\N_0}\Gmet_n$ with respect to $d_G$. Then  $\Omega_{d_G}$ can be naturally and homeomorphically identified with $X$ in such a way that it is bi-Lipschitz equivalent to both $R$ and the Euclidean metric.
\end{remark}

\begin{remark} 
The metric $d_G$ is partially self-similar on $X$ in that $$d_G(F_w(x),F_w(y)) =\left(\frac{1-\alpha}{2}\right)^nd_G(x,y)$$ for any $w\in\A^n$. To see this, note that for any $k>0$ there is a bijection between paths in $\Gmet_k$ from $x$ to $y$, and paths in $F_w(\Gmet_k)\subset \Gmet_{k+n}$. It is easy to see that a minimizing path will not leave $F_w(\Gmet_k)$, and so this implies that $d_{n+k}(F_w(x),F_w(y))= \left(\frac{1-\alpha}{2}\right)^nd_k(x,y)$.  Self-similarity follows by passing to the limit. 
In addition, from the following picture 
\begin{figure}[h!tpb]
\begin{tikzpicture}[scale=0.55]
\coordinate [label=left: {$
p_1$}](p_1) at (0,0);
\fill (p_1) circle (2.5pt);
\coordinate [label=above: {$p_2$}] (p_2) at (3.5,6.0621);
\fill (p_2) circle (2.5pt);
\coordinate (p_3) at (7,0);
\coordinate (p_12) at (1.5,2.59805);
\coordinate [label=left: {$F_2(p_1)$}](p_21) at (2,3.46405);
\fill (p_21) circle (2.5pt);
\coordinate (p_13) at (3,0);
\coordinate (p_31) at (4,0);
\coordinate [label=right:{$F_2(p_3)$}] (p_23) at (5,3.46405);
\fill (p_23) circle (2.5pt);
\coordinate [label=below left:{$p_4$}] (p_4) at (5.25,3);
\fill (p_4) circle (2.5pt);
\coordinate (p_32) at (5.5,2.59805);
\draw (p_1) -- (p_13) node[midway, above] {\small{$F_1(X)$}};
\draw (p_21) -- (p_23) node[midway, above] {\small{$F_2(X)$}};
\draw (p_31) -- (p_3) node[midway, above] {\small{$F_3(X)$}};
\draw (p_12) -- (p_13)  (p_32) -- (p_31)   (p_23) -- (p_21) (p_13) -- (p_31);
\draw (p_1) -- (p_2) -- (p_3) -- cycle;
\end{tikzpicture}
\end{figure}
one can conclude that the geodesic diameter of $X$ is the distance from 
$p_1$ to $p_4$. Here $p_4$ is the fixed point of $F_4$, i.e. the midpoint of the line segment connecting $p_2$ and $p_3$.
\end{remark}

\begin{remark}
Since the Euclidean metric, $d_G$ and $R$ are all equivalent metrics, the Hausdorff dimension of $X$ with respect to any of these metrics is the same value, in particular 
$$\dim(X)=\max\left\{1,\ \frac{\ln3}{\ln2-\ln(1-\alpha)}\right\}.
$$ 
\end{remark}

\subsection{Proof of Theorem~\ref{thm resistance form}}
This subsection proves Theorem~\ref{thm resistance form} using the results of Section~\ref{subsec metric obs}. The subsections \ref{subsubsec:traces} and \ref{subsubsec:resistanceform} establish that $\eng$ can be extended to a resistance form on $X$ using techniques from \cite{Kig12}. In Subsection \ref{subsubsec:dom_characterization} it is shown that the domain of this resistance form is $H^1(X)$.  

\subsubsection{Finite dimensional resistance forms on $X$}\label{subsubsec:traces}
The first part of the proof relies on~Theorem~\ref{lemR}~\eqref{lemR1}. Here we do not provide the domain $\Dom\eng$ of $\eng$ explicitly but instead use an abstract result of 
Kigami concerning compatible sequences of resistance forms~\cite[Theorem 3.13]{Kig12}. 

For any nonempty finite subset $S\subset X$ and any function $u\in\ell(S)$ we define 
\begin{equation*}\label{eES}
\eng_S(u,u)=\inf_n \inf_v\set{\Eng_n(v,v):v\big|_S=u, v\in\F_n}.
\end{equation*}
As a consequence of Theorem~\ref{lemR}~\eqref{lemR1},
we have the following facts: 
Each biniliarized form $\eng_S(u,v)$ is a resistance form on the finite set $S$, and in particular $\eng_S$ vanishes only on constants. 

(RF1) Clearly $\ell(S)$ is a linear subspace of itself and if $u\equiv\const$ then $\eng_S(u,u)=0$. Conversely, if $\eng_S(u,u)=0$ then also $u$ is constant because if $u$ were nonconstant, there would be $x\neq y\in S$ with $u(x)\neq u(y)$ and thus
\[
\eng_n(v,v)\geq \frac{|u(x)-u(y)|^2}{d_n(x,y)}>0
\]
for all $v\in\F_n$ with $v\vert_S\equiv u$ by Lemma~\ref{engestimate}, implying $\eng_S(u,u)>0$.

(RF2) $(\ell(S)/_{\sim},\eng_S^{1/2})$ is Hilbert because $\ell(S)$ is finite dimensional, and thus (R1) implies $\eng_S$ is an inner product on the quotient space. 

(RF3) $\ell(S)$  separates points.

(RF4) Define $R_S(p,q):=\sup\big\{\frac{|u(p)-u(q)|^2}{\eng_S(u,u)}\colon~u\in\ell(S),~\eng_S(u,u)\neq 0\big\}$ for any $p,q\in S$.
\begin{align*}
R_S(p,q)&=\sup_{u\in\ell(S)}\frac{|u(p)-u(q)|^2}{\inf\limits_n\inf\limits_{\substack{v\in\F_n\\v|_S=u}}\eng_n(v,v)}=\sup_{u\in\ell(S)}\sup_n\sup\limits_{\substack{v\in\F_n\\v|_S=u}}\frac{|u(p)-u(q)|^2}{\eng_n(v,v)}\\
&=\sup_n\sup_{u\in\ell(S)}\sup\limits_{\substack{v\in\F_n\\v|_S=u}}\frac{|v(p)-v(q)|^2}{\eng_n(v,v)}
{=\sup_n\sup_{v\in\F_n}\frac{|v(p)-v(q)|^2}{\eng_n(v,v)}}\\
&=\sup_n R_n(p,q)=R(p,q)<\infty,
\end{align*}
where $R(p,q)$ was defined in Theorem~\ref{lemR}~\eqref{lemR1}. Interchanging of suprema is possible because 
 $\sup_n\sup\{\frac{|u(p)-u(q)|^2}{\eng_n(v,v)}~\colon v\in\F_n,~v|_S=u\}$ is uniformly bounded by $d_G(p,q)$ by Lemma~\ref{engestimate}.

(RF5) Consider $u\in\ell(S)$ and $\overline{u}:=0\vee u\wedge 1$. On the one hand, $\eng_n(\overline{v},\overline{v})\leq \eng_n(v,v)$ for any $v\in\F_n$, which implies
\[
\inf\set{\eng_n(\overline{v},\overline{v})~\colon~v\in\F_n,~v|_S=u}\leq\inf\set{\eng_n(v,v)~\colon~v\in\F_n,~v|_S=u}=\eng_S(u,u).
\]
On the other hand, for all $n\geq 1$, $\{\overline{v}~\colon~v\in\F_n,~v|_S=u\}\subseteq\{v~\colon~v\in\F_n,~v|_S=\overline{u}\}$ 
and thus
\[
\eng_S(\overline{u},\overline{u})=\inf\set{\eng_n(v,v)~\colon~v\in\F_n,~v|_S=\overline{u}}\leq\inf\set{\eng_n(\overline{v},\overline{v})~\colon~v\in\F_n,~v|_S=u}.
\]

\subsubsection{Compatible sequences of finite dimensional resistance forms}\label{subsubsec:resistanceform}
In this section, we prove that the bilinear form $\eng$ can be extended to a resistance form. To do this, we show that the family of resistance forms $\set{\eng_S,S\subset X}$ is compatible in the sense of~\cite[Def. 3.12]{Kig12}. For a sequence of finite sets satisfying $S_1\subset S_2\subset\ldots\subset S_k\subset\ldots$ we prove that for any $k\in\N$ and any $u\in\ell(S_k)$
\begin{equation*}\label{eq comp sequence}
\eng_{S_k}(u,u)=\inf\set{\eng_{S_{k+1}}(v,v)~:~v\in\ell(S_{k+1}),~v|_{S_k}=u}.
\end{equation*}
Indeed,
\begin{align*}
\inf\limits_{\substack{v\in\ell(S_{k+1})\\ v|_{S_k}=u}}\eng_{S_{k+1}}(v,v)&=\inf\limits_{\substack{v\in\ell(S_{k+1})\\ v|_{S_k}=u}}\inf\limits_n\inf\limits_{\substack{w\in\F_n\\ w|_{S_{k+1}}=v}}
 \eng_n(w,w)\\
 &=\inf\limits_n\inf\limits_{\substack{v\in\ell(S_{k+1})\\ v|_{S_k}=u}}\inf\limits_{\substack{w\in\F_n\\ w|_{S_{k+1}}=v}}\eng_n(w,w)
 {=\inf_n\inf\limits_{\substack{w\in\F_n\\ w|_{S_k}=u}}\eng_n(w,w)}=\eng_{S_k}(u,u).
\end{align*}

From~\cite[Theorem 3.13]{Kig12} we obtain the existence of a resistance form $(\eng',\Dom\eng')$ given by $\eng'(u,u)=\lim_{k\to\infty}\eng_{S_k}(u|_{S_k},u|_{S_k})$. This is a resistance form on the closure of $\bigcup_k S_k$ w.r.t. the effective resistance metric $R'$ of $\eng'$. From the proof of (RF4) we have that the metrics $R$ and $R'$ coincide on $S_k$ for any $k$. Since the sequence $(S_k)_k$ converges to a dense set in $X$ with respect to $R$ and $X$ is complete, the completion of $\cup_k S_k$ with respect to $R'$ is the completion with respect to $R$. In particular, $(\eng',\Dom\eng')$ is a resistance form on $X$ and $R' = R$ on all of $X$. 

In order to show that $\eng'$ is an extension of $\eng$, we prove that $\eng'(u) = \eng_k(u)$ for all $u\in\F_k$.  Without loss of generality, choose
\[
S_k:=\left\{\Phi_e\left(\frac{r(e)m}{2^k}\right)~\colon~0\leq m\leq 2^k,\,e\in J_k\right\}.
\]
This choice is important because $S_k$ becomes dense in a uniform way.

For any $u\in \F_n$ and $n\in\N$, $u\in \dom \eng'$ because $\eng_{S_k} (u|_{S_k},u|_{S_k}) \leq \eng_n(u,u)$ for all $k$, hence $\eng'(u,u) \leq \eng_n(u,u) = \eng(u,u)$. On the other hand, $\eng(v,v)$ with $v|_{S_k}=u|_{S_k}$ is minimized by a function $u_k$ such that $ u_k \circ \Phi_e: I_e \to \R$ is a piecewise linear function that interpolates between values of $u$ on points in $\Phi^{-1}_e(S_k)$. 
Further, since $S_k$ includes the endpoints of $I_e$, the function which extends these values of $v_k$ to $J_n^c$ by constants is well defined and it will be the minimizer. In particular, $v_k\in\F_n$ and it is constant on all edges where $u$ is constant. Thus,
\[
\eng'(u,u) =\lim_{k\to\infty} \eng_{S_k}(u|_{S_k},u|_{S_k})= \lim_{k\to\infty} \eng_k(v_k,v_k) = \eng(u,u)
\]
because the points in $S_k$ become uniformly dense in $\Phi_e(I_e)$ and hence
 \[
\lim_{k\to\infty} \int_0^{\ell_k} ((v_k\circ\Phi_e)'(t))^2 \ dt = \int_0^{\ell_k} (u\circ\Phi_e)'(t))^2 \ dt.
 \]
%

This implies that $\eng(u,u) = \eng'(u,u)$ for any $u\in \cup_n \F_n$, so that the resistance form $\eng'$ is an extension of $\eng$, and from now on we shall refer to $(\eng',\dom\eng')$ as $(\eng,\dom \eng)$. Note that $\cup_{n=1}^\infty\F_n\subsetneq \dom\eng$. This also implies that the construction is in fact independent on which dense countable subset is chosen.

Moreover, it was shown above that for any $v\in\ell(S_k)$ and any $k$, there is some $n$ and $u\in \F_n$ such that $\eng_{S_k}(v,v) = \eng_n(u,u)$ and $u(x)= v(x)$ for all $x\in S_k$.  Thus, by~\cite[Theorem 3.13]{Kig12} and property (RF4) of resistance forms, a function $u\in \ell(X)$ is in $\Dom\eng$ if and only if there exists a sequence $(u_n)_n$, $u_n\in\F_n$ such that
\[
\norm{u_n-u}_{\infty}\to 0\qquad\text{and}\qquad (u_n)_n\text{ is }\eng-\text{Cauchy}.
\]

\subsubsection{Characterization of \ $\Dom\eng$}\label{subsubsec:dom_characterization}
After having established that $\eng$ can be extended to a resistance form, the final step in proving Theorem~\ref{thm resistance form} is showing that $\dom\eng = H^1(X)$. This requires the full strength of Theorem~\ref{lemR} and approximation by quantum graphs.  

In particular, we get that $u\in H^1(X)$ if and only if there is a sequence $u_n\in\F_n$ such that 
\[
u_n\to u\text{\ uniformly on $X$}
\]
and 
\[
(u_n)_n\text{\ is an $\eng$-Cauchy sequence.} 
\]
Then we have 
\begin{equation*}\label{eng as limit}
\eng(u,u)=\lim\eng(u_n,u_n)<\infty, 
\end{equation*}
where $\eng$ is defined in~\eqref{e-one}. 
To see that $H^1(X)$ is closed under the above type of convergence, consider a sequence $(u_n)_n\subset H^1(X) \subset C(X)$ is given. Its pointwise limit $u$ belongs to $H^1(X)$ because on any interval $I$ contained in $X$, $u_n$ restricted to that interval will be in the classical Sobolev space $H^1(I)$, which is closed under the above limits.

It is easy to see that $\F_n\subset H^1(X)$ for all $n$, and because $\dom \eng$ is the closure of $\cup_n\F_n$ under the above kind of limit, this implies that $\dom \eng\subset H^1(X)$.

Given a function $u\in H^1(X)$, we construct and $\eng$-Cauchy sequence $(u_n)_n$ with $u_n\in \F_{n}$ and $u_n \to u$ uniformly. Without loss of generality, we can assume that $u$ is linear on all straight line segments $\Phi_e(I_e)$ because the energy orthogonal complement of such functions are those which vanish at all the endpoints of $\Phi_e(I_e)$, and such functions are easily approximated by elements of $\F_*$. If $u$ is linear on each straight line segment in $\Phi_e(I_e)$ and $n\in\N$ is fixed, we approximate $u$ by averaging on the cells $F_w(X)$ for $w\in\A^n$ and interpolating linearly on the segments $\Phi_e(I_e)$, $e\in J_n$. 
It is elementary to prove that such as sequence $(u_n)_n$ is an $\eng$-Cauchy sequence. The key observation is that, according to Theorem~\ref{lemR}, the effective resistance diameter of the cells $F_w(X)$ is controlled by the resistance of the segments $\Phi_e(I_e)$ for $e\in J_n$, and so the energy of the difference between $u$ and $u_{n+1}$ is controlled by the energy of $u$ contained inside the cells $F_w(X)$, which vanishes as $n\to\infty$. 

This concludes the proof of Theorem \ref{thm resistance form}. 

\subsection{Approximation by quantum graphs $\Gmet_n$}

Having proven Theorem~\ref{thm resistance form} we end this section with another useful characterization of $\dom\eng = H^1(X)$.

\begin{proposition}
A function $u\in C(X)$ belongs to $H^1(X)$ if and only if the restriction of $u$ to any $\Gmet_n$ is a finite energy function and
\[
\sup_n \eng_{\Gmet_n}(u\big|_{\Gmet_n},u\big|_{\Gmet_n}) < \infty.
\]
In this case, the sequence $\eng_{\Gmet_n}(u\big|_{\Gmet_n},u\big|_{\Gmet_n})$ is non-decreasing and
\[
\eng(u,u)=\lim_{n\to\infty}\eng_{\Gmet_n}(u\big|_{\Gmet_n},u\big|_{\Gmet_n})=\sup_n \eng_{\Gmet_n}(u\big|_{\Gmet_n},u\big|_{\Gmet_n}).
\]
\end{proposition}

\begin{proof}
On one hand, if $u\in H^1(X)$, then $|u(p)-u(q)|\leq\eng(u,u)R(p,q)$ by (RF4) and Remark~\ref{Eu bilip R}. Therefore, $u$ is continuous with respect to the effective resistance and hence continuous with respect to the Euclidean metric by Theorem \ref{lemR}~\eqref{i3}, i.e. $H^1(X)\subset C(X)$. Furthermore, it is easy to see that $\eng_{\Gmet_n}(u\big|_{\Gmet_n},u\big|_{\Gmet_n})\leq \eng(u,u)$, as the latter is the sum over a set of positive terms and the former is the sum over a subset of these terms. 


On the other hand, if $u\in C(X)$ with $\sup_n \eng_{\Gmet_n}(u\big|_{\Gmet_n},u\big|_{\Gmet_n}) < \infty,$ it can be seen that $\lim_n \eng_{\Gmet_n}(u\big|_{\Gmet_n},u\big|_{\Gmet_n}) = \eng(u,u)$ by observing that for every edge $e$ of $X$ there is $n_0\in\N$ such that $e$ is a subset of a an edge of $\Gmet_n$ for all $n> n_0$ and  thus the limit of $ \eng_{\Gmet_n}(u\big|_{\Gmet_n},u\big|_{\Gmet_n})$ is a rearrangement of the sum $\eng(u,u)$. 
\end{proof}

\section{Spectral asymptotics}\label{sec-asymp}

We know from~\cite[Chapter 9]{Kig12} that a resistance form together with a locally finite re\-gu\-lar measure induces a Dirichlet form on the corresponding $L^2$-space. By introducing an appropriate measure $\mu$ on $X$, we can therefore obtain a Dirichlet form and a Laplacian on $L^2(X,\mu)$. The spectral properties of this operator strongly depend on the measure, that we choose in a weakly self-similar manner in view of the geometric properties of $X$.


Recall from the introduction the parameters $\beta\in(0,1/3)$ and $s=\frac{1-3\beta}{3}$. For any $w\in\A^n$, $n\in\N$, we define
\[
\mu(F_w(U)):=s^{\abs{w}}\mu(U),
\]
where $\mu$ is Lebesgue measure on $I_e$ for $e\in J_1$ with $\mu(I_e) = \beta$. Notice that $\beta$ and $s$ are related in such a way that $\mu(X)=1$.

As a direct consequence of Theorem~\ref{lemR}, $X$ is compact with respect to the resistance me\-tric and it follows from~\cite[Corollary 6.4]{Kig12} that the induced Dirichlet form coincides with $(\eng, \Dom \eng)$. Next definition is a well-known fact from the theory of Dirichlet forms that can be found in~\cite[Corollary 1.3.1]{FOT11}. 


\begin{definition}
The \textit{Laplacian associated with} $(\eng, \Dom \eng)$ is the unique non-negative self-adjoint operator $\Delta_\mu\colon\Dom\Delta_\mu\to L^2(X,\mu)$ such that $\Dom\Delta_\mu$ is dense in $L^2(X,\mu)$ and 
\[\eng(u,v)=-\int_X \Delta_\mu u\cdot v\,d\mu\qquad \forall\,v\in\Dom \eng.\]
\end{definition}


Recall that $r:=\frac{1-\alpha}{2}$ denotes the scaling factor of the similitudes $F_1,F_2,F_3$ and write $I_e=[0,r^n\alpha]$ for any $e\in J_{n+1}\setminus J_n$, $n\in\N_0$.
\begin{lemma}\label{lem-r}
For any $u\in\Dom\eng$,
\[\eng(u,u)=\sum_{i=1}^3r^{-1}\eng(u\circ F_i,u\circ F_i)+\sum_{e\in J_1}\int_0^{\alpha}\abs{u'}^2dx.\]
\end{lemma}
\begin{proof}
Let $u\in H^1(X)$. 
\begin{align*}
\eng(u,u)&=\sum_{k=1}^\infty\sum_{e\in J_k\setminus J_{k-1}}\int_0^{r^{k-1}\alpha}\abs{u'}^2dx\\
&=\sum_{i=1}^3\sum_{k=1}^{\infty}\sum_{e\in F_i(J_k\setminus J_{k-1})}\int_0^{r^k\alpha}\abs{u'}^2dx+\sum_{e\in J_1}\int_0^{\alpha}\abs{u'}^2dx
\end{align*}
Applying the transformation of variables $x=F_i (y)$ we get that
\begin{align*}
\eng(u,u)&=\sum_{i=1}^3\sum_{k=1}^{\infty}\sum_{e\in F_i(J_k\setminus J_{k-1})}r^{-1}\int_0^{r^{k-1}\alpha}\abs{(u\circ F_i)'}^2dy+\sum_{e\in J_1}\int_0^{\alpha}\abs{u'}^2dx\\
&=\sum_{i=1}^3r^{-1}\eng(u\circ F_i,u\circ F_i)+\sum_{e\in J_1}\int_0^{\alpha}\abs{u'}^2dx.
\end{align*}
\end{proof}

By iterating we get the following generalization of this Lemma.

\begin{corollary}\label{corollary: energy identity}
For any $u\in\Dom\eng$ and $m\in\N$,
\[\eng(u,u)=\sum_{w\in\Am}r^{-m}\eng(u\circ F_w,u\circ F_w)+\sum_{k=0}^{m-1}r^{-k}\sum_{w\in\A^k}\sum_{e\in J_1}\int_0^{\alpha}\abs{(u\circ F_w)'}^2dx.\]
\end{corollary}

The \textit{eigenvalue counting function} of $\Delta_{\mu}$ subject to Neumann (resp. Dirichlet) boundary conditions is defined as
\[N_N(x):=\#\{\lambda\;\text{Neumann eigenvalue of }\Delta_{\mu}~\colon~\lambda\leq x\},\]
respectively
\[N_D(x):=\#\{\lambda\;\text{Dirichlet eigenvalue of }\Delta_{\mu}~\colon~\lambda\leq x\}\]
counted with multiplicity. 
In our particular case, the boundary of $X$ is the set $V_0=\{p_1,p_2,p_3\}$.

This function can also be defined for Dirichlet forms by considering that $\lambda\in\R$ is an eigenvalue of $\eng$ if and only if there exists $u\in\Dom\eng$ such that $\eng(u,v)=\lambda\int_X uv\,d\mu$ $\forall\,v\in\Dom\eng$. In this case the eigenvalue counting function
\[N(x;\eng,\Dom\eng):=\#\{\lambda\;\text{eigenvalue of }\eng~\colon~\lambda\leq x \}\]
coincides with $N_N(x)$ (see~\cite[Proposition 4.1]{Lap91}). Analogously it holds that
\[N_D(x)=N(x;\eng^0,\Dom \eng^0),\] 
where $\Dom \eng^0:=\{u\in\Dom\eng~\colon~u_{\vert_{V_0}}\equiv 0\}$ and $\eng^0:=\eng_{\vert_{\Dom \eng^0\times\Dom \eng^0}}$.

The asymptotic behaviour of the eigenvalue counting function is described by the so-called \textit{spectral dimension} of $X$, that is the non-negative number $d_S$ such that
\begin{equation*}\label{eq: def d_S  X}
\lim_{x\to\infty}N_{N/D}(x) x^{-d_S /2}=C<\infty.
\end{equation*}

The expression $N_{N/D}(x)$ means that a property holds for both $N_N(x)$ and $N_D(x)$ and we will use it in the following to simplify notation.

The main result of this section is Theorem~\ref{theorem: d_S X}, which indicates the value of the spectral dimension of $X$. The proof of this theorem is divided into several lemmas that estimate the eigenvalue counting functions $N_N(x)$ and $N_D(x)$ and it mainly follows ideas of~\cite{Kaj10}, that can be applied due to the choice of the measure $\mu$.

We introduce the norm on $\Dom\eng$ given by
\[
\norm{u}_{\eng^{(1)}}:=\left(\eng(u,u)+\norm{u}_{L^2(X,\mu)}^2\right)^{1/2}.
\]

\subsection*{Upper bound}

Let us write $X_w:=F_w(X)$ for each $w\in\A^n$, $n\in\N$, and define $\Xm :=\bigcup_{w\in\Am}X_w$ and $I_m:=X\setminus \Xm $ for each $m\in\N$. 

On the one hand, we consider the pair $(\eng_{I_m},\Dom\eng_{I_m})$ given by
\begin{equation}\label{def eng_I_m}
\left\{\begin{array}{l}
\Dom\eng_{I_m}:= \bigoplus\limits_{e\in J_m}H^1(I_e,\mu_{\vert_{I_e}}),\\
\eng_{I_m}(u):=\sum\limits_{e\in J_m}\int_{I_e} ((u\circ\phi_e)')^2 dx,
\end{array}\right.
\end{equation}
which is a Dirichlet form on an $L^2$ space that can be identified with $\bigoplus\limits_{e\in J_m}L^2(I_e,\mu_{\vert_{I_e}})$.

On the other hand, we consider the Dirichlet form $(\eng_{\Xm },\Dom\eng_{\Xm })$ in $L^2(\Xm ,\mu_{\vert_{\Xm }})$ constructed following Section~\ref{sect HanoiAtts} and Section~\ref{resistance}, substituting $X$ by $\Xm $.

\begin{lemma}\label{lemma: Neum decomp}
For each $m\in\N$
\[
N_N(x)\leq N(x;\eng_{\Xm },\Dom\eng_{\Xm })+N(x;\eng_{I_m},\Dom\eng_{I_m})\qquad\forall\, x\geq 0.
\]
\end{lemma}
\begin{proof}
Since $\Dom\eng\subseteq\Dom\eng_{\Xm }\oplus\Dom\eng_{I_m}$, the minimax principle yields 
\[
N_N(x)=N(x;\eng,\Dom\eng)\leq N(x;\eng,\Dom\eng_{\Xm }\oplus\Dom\eng_{I_m}).
\] 
The assertion now follows from~\cite[Proposition 4.2]{Lap91} and~\cite[Lemma 4.2]{Lap91}. Note that in this proof we first consider $\eng_{\Xm }$ and $\eng_{I_m}$ as bilinear forms in $L^2(X,\mu)$ and afterwards each of them is considered on $L^2(\Xm ,\mu_{\vert_{\Xm }})$ and $L^2(I_m,\mu_{\vert_{I_m}})$ respectively.
\end{proof}

\begin{lemma}\label{lemma: upper bound part I}
For each $m\in\N$ and each $L$ subspace of $\Dom\eng_{\Xm }$, define
\begin{align*}
\lambda(L)&:=\sup\{\eng_{\Xm }(u,u)~\colon~u\in L,\;\norm{u}_{L^2(\Xm ,\mu_{\vert_{\Xm }})}=1\},\\
\lambda_n &:=\inf\{\lambda(L)~\colon~L\subseteq\Dom\eng_{\Xm },\,\dim L=n\}.
\end{align*}
Then, it holds that
\[\lambda_{3^m+1}\geq C_P(rs)^{-m}.\]
\end{lemma}
\begin{proof}
By Corollary~\ref{corollary: energy identity} and the definition of $\eng_{\Xm }$ we have that 
\begin{equation}\label{eq: self-type equality}
\eng_{\Xm }(u,u)=\sum_{w\in\Am}r^{-m}\eng(u\circ F_w,u\circ F_w)
\end{equation}
for all $u\in\Dom\eng_{\Xm }$. Note that all of the components of the above sum are positive.

We follow a similar argument as in~\cite[Lemma 4.5]{Kaj10}, which is included for completeness: consider $L_0:=\{\sum_{w\in\Am}a_w1_{X_w}~\colon~a_w\in\R\}$. This is a $3^m-$dimensional subspace of $\Dom\eng_{\Xm }$ such that $\eng_{\Xm }\vert_{L_0\times L_0}\equiv 0$. For a $(3^m+1)-$dimensional subspace $L\subseteq\Dom\eng_{\Xm }$, we consider the finite-dimensional subspace of $\Dom\eng_{\Xm }$ given by $\tilde{L}:=L_0+L$. The non-negative self-adjoint operator associated with $\eng\vert_{\tilde{L}\times\tilde{L}}$ may be expressed by a matrix $A$ whose $3^m+1-$th smallest eigenvalue is given by
\[\lambda_A:=\inf\{\lambda(L')~\colon~L'\subseteq \widetilde{L},\;\dim L'=3^m+1 \}.\]
Call $u_A$ the corresponding eigenfunction, renormalized so that $\int_{\Xm }u_A^2d\mu=1$. Since $(\eng,H^1(X))$ is a resistance form on $X$, the associated resistance metric $R$ is compatible with the original topology of $X$ by Theorem~\ref{lemR}, and $u_A$ is orthogonal to $L_0$, a uniform Poincar\'e inequality (see~\cite[Definition 4.2]{Kaj10} for the self-similar case) holds for $u_A$. This together with equality~\eqref{eq: self-type equality} leads to 
\begin{align*}
\lambda_A&=\lambda_{3^m+1}\geq \eng_{\Xm }(u_A,u_A) = 
\sum_{w\in\Am}r^{-m}\eng(u_A\circ F_w,u_A\circ F_w)\\
&\geq C_P\sum_{w\in\Am}\frac{r^{-m}}{\mu(X_w)}\int_{X_w} \abs{u_A}^2d\mu
=
 \frac{r^{-m}C_P}{\mu(X_w)}\sum_{w\in\Am}\int_{X_w} \abs{u_A}^2d\mu 
 = 
  \frac{C_P}{(rs)^{m}},
\end{align*}
where $$C_P\geq \frac1{\Diam\limits_R(X)}$$ 
is the constant of the  Poincar\'e inequality. 
Note that here $u_A$ is a function orthogonal to all locally constant functions on $\Xm $.
\end{proof}

\begin{lemma}\label{lemma: upper bound part II}
There exist a constant $\widetilde{C}>0$ and $x_0>0$ such that
\begin{itemize}
\item[(i)] if $0<rs<\frac{1}{9}$, then
\[N_N(x)\leq\widetilde{C} x^{1/2}+o(x^{1/2}),\]
\item[(ii)] if $rs=\frac{1}{9}$, then
\[N_N(x)\leq\widetilde{C}  x^{1/2}\log x,\]
\item[(iii)] if $\frac{1}{9}<rs<\frac{1}{6}$, then
\[N_N(x)\leq\widetilde{C} x^{\frac{\log 3}{-\log (rs)}}+o(x^{\frac{\log 3}{-\log (rs)}}),\]
\end{itemize}
for all $x\geq x_0$.
\end{lemma}
\begin{proof}
Let $x_0:=4\pi^2s^3r^3$. For any $x>x_0$ we can find $m\geq 1$ such that
\begin{equation}\label{eq: 1.choice of m}
\frac{4\pi^2}{\alpha\beta(sr)^{m-4}}\leq x<\frac{4\pi^2}{\alpha\beta(sr)^{m-3}}.
\end{equation}
By Lemma~\ref{lemma: upper bound part I} we know that 
\[\lambda_{3^m+1}\geq \frac{C_P}{(sr)^{m}}\geq x\]
and hence 
\begin{equation}\label{eq: asym X_m}
N(x;\eng_{\Xm },\Dom\eng_{\Xm })\leq 3^m\leq C_1x^{\frac{\ln 3}{-\ln (rs)}}
\end{equation}
for $C_1=3^4\left(\frac{\alpha\beta}{4\pi^2}\right)^{\frac{-\ln 3}{\ln (rs)}}$.

On the other hand, since $I_m$ is the disjoint union of $1-$dimensional intervals, $$N(x;\eng_{I_m},\Dom\eng_{I_m})=\sum_{e\in J_m}N_e(x),$$ where $N_e(x)$ denotes the eigenvalue counting function of the Laplacian on $L^2(I_e,\mu_{\vert_{I_e}})$, that we denote by $\Delta_{\mu_{\vert_{I_e}}}$.

Without loss of generality, let us consider $I_e=[0,r^k\alpha]$ and suppose that $\lambda$ is an eigenvalue of the Laplacian $\Delta_{\mu_{\vert_{I_e}}}$ with eigenfunction $f\in H^1(I_e,\mu_{\vert_{I_e}})$. Then,
\begin{equation*}
\int_0^{r^k\alpha} f'g'\,dx=\lambda\int_0^{r^k\alpha}fg\,d\mu=\lambda\frac{\mu(I_e)}{m(I_e)}\int_0^{r^k\alpha}fg\,dx
\end{equation*}
for all $g\in H^1(I_e)$, where $m(I_e)$ denotes the Lebesgue measure of $I_e$. This means, $\lambda\frac{\mu(I_e)}{m(I_e)}$ is an eigenvalue of the classical Laplacian $\Delta$ on $L^2(I_e,dx)$ subject to Neumann boundary conditions. The converse holds by the same calculation, so we can say that $N_e(x)=N_N^{I_e}\left(\frac{x\mu(I_e)}{m(I_e)}\right)$ for all $x\geq 0$. Here $N^{I_e}_N(\cdot)$ denotes the eigenvalue counting function of the classical Laplacian on $L^2(I_e,dx)$ subject to Neumann boundary conditions.

From Weyl's theorem for the asymptotics of the eigenvalue counting function for the classical Laplacian on bounded sets of $\R$ (see~\cite{Wey12}), we know that
\[
N_e(x)=\frac{(\mu(I_e)m(I_e))^{1/2}}{2\pi}(x^{1/2}+O(1))=\frac{(\alpha\beta)^{1/2}(rs)^{m/2}}{2\pi}(x^{1/2}+O(1)),
\]
hence 
\begin{equation}\label{eq: J_m EVCF}
N(x;\eng_{I_m},\Dom\eng_{I_m}) = \sum_{n=1}^m\frac{(\alpha\beta)^{1/2}(9rs)^{n/2}}{2\pi}(x^{1/2}+O(1)),
\end{equation}
which is the counting function of the set
\[
\bigcup_{n=0}^m \set{\frac{(2\pi k)^2}{\alpha\beta (rs)^n}~|~k=1,2,\ldots}.
\] 
If $0<rs<1/9$, this expression is a convergent geometric series bounded by a constant so we get from~\eqref{eq: J_m EVCF} that
\[
N(x;\eng_{I_m},\Dom\eng_{I_m})=O(x^{\frac{1}{2}}).
\]
Since $\frac{\ln3}{-\ln rs}<\frac{1}{2}$, Lemma~\ref{lemma: Neum decomp} leads to $(i)$.

Now, note that~\eqref{eq: 1.choice of m} is equivalent to
\[
m\leq \frac{\log x}{-\log (rs)}+C<m+1,
\]
where $C=\frac{\log(\alpha\beta)-2\log(2\pi)}{-\log (rs)}$ and so we have that
\begin{equation*} 
\sum_{n=1}^m (9rs)^{n/2}\leq \int_0^{ \frac{\log x}{-\log (rs)}+C} (9rs)^{y/2}dy.
\end{equation*}
If $rs=\frac{1}{9}$, then $9rs=1$ and the integral becomes $\frac{\log x}{-\log (rs)}+C$. Moreover, $\frac{\ln 3}{-\ln (rs)}=\frac{1}{2}$, hence~\eqref{eq: asym X_m} and~\eqref{eq: J_m EVCF} lead to $(ii)$. Finally, if $\frac{1}{9}<rs<\frac{1}{6}$, we have that
\begin{align*}
\sum_{n=1}^{m-4}(9rs)^{\frac{n}{2}}&=\sum_{n=0}^{m-4}(9rs)^{\frac{n}{2}}-1=\frac{1-(9rs)^{\frac{m-3}{2}}}{1-9rs}-1=\frac{9rs-(9rs)^{\frac{m-3}{2}}}{1-9rs}\\
&=\frac{3(rs)^{1/2}}{9rs-1}((9rs)^{\frac{m-4}{2}}-3(rs)^{1/2})\\
&\leq\frac{3(rs)^{1/2}}{9rs-1}\left(x^{\frac{-\ln 9rs}{2\ln (rs)}}\cdot \left(\frac{\alpha\beta}{4\pi^2}\right)^{\frac{-\ln 9rs}{2\ln rs}}-3(rs)^{1/2}\right)
\end{align*}
and hence
\begin{align*}
N(x;\eng_{I_m},\dom\eng_{I_m})&\leq
\frac{\left(\frac{4\pi^2 rs}{3^6\alpha\beta}\right)^{1/2}}{9rs-1}C_1x^{\frac{\ln 3}{-\ln (rs)}}+O(x^{1/2})
\end{align*}
Lemma~\ref{lemma: Neum decomp} finally leads to
\[
N(x,\eng,\dom\eng)\leq \widetilde{C} x^{\frac{\log 3}{-\log (rs)}}+o(x^{\frac{\log 3}{-\log (rs)}}),
\]
with $\widetilde{C}=2\max\left\{1,\left(\frac{4\pi^2 rs}{3^6\alpha\beta}\right)^{\frac{1}{2}}\cdot\frac{1}{9rs-1}\right\}C_1$.
\end{proof}

\subsection*{Lower bound}

Recall that $(\eng^0,\Dom\eng^0)$ is the Dirichlet form whose associated non-negative self-adjoint operator is the Laplacian $\Delta_{\mu}$ subject to Dirichlet boundary conditions. Let us now write for each $w\in\A^n$, $n\in\N$, $X_w^0:=F_w(X\setminus V_0)$, and $\Xm ^0:=\bigcup_{w\in\Am} X_w^0$ for each $m\in\N$. Since $\Xm ^0$ is open, we know from~\cite[Theorem 4.4.3]{FOT11} that the pair $(\eng_{\Xm ^0},\Dom\eng_{\Xm ^0})$ given by
\begin{equation*}
\begin{cases}\begin{aligned}
\Dom\eng_{\Xm ^0}&:=\overline{\{u\in\Dom\eng\,\vert\;\supp(u)\subseteq \Xm ^0\}},\\
\eng_{\Xm ^0}&:=\eng\vert_{\Dom\eng_{\Xm ^0}\times\Dom\eng_{\Xm ^0}},&
\end{aligned}
\end{cases}
\end{equation*}
where the closure is taken with respect to $\norm{\cdot}_{\eng^{(1)}}$, is a Dirichlet form on $L^2(\Xm ^0,\mu_{\vert_{\Xm ^0}})$. Analogously, we define for each $w\in\A^n$, $n\in\N$, the Dirichlet form $(\eng_{X_w^0},\Dom\eng_{X_w^0})$ on $L^2(X_w^0,\mu_{\vert_{X_w^0}})$. Moreover, we consider $(\eng_{I_m},\Dom\eng_{I_m}^0)$ where $\eng_{I_m}$ is defined as in~\eqref{def eng_I_m} and $\Dom\eng_{I_m}^0:=\bigoplus\limits_{e\in J_m}H^1_0(I_e,\mu_{\vert_{I_m}})$.

\begin{lemma}\label{lemma: Dirichlet decomp}
For each $m\in\N$ and $x\geq 0$,
\begin{align*}
N_D(x)&\geq
\sum_{w\in\Am} N(x;\eng_{X_w^0},\Dom\eng_{X_w^0})+N(x;\eng_{I_m},\Dom\eng_{I_m}^0).
\end{align*}
\end{lemma}
\begin{proof}
See~\cite[Lemma 4.8]{Kaj10}.
\end{proof}

\begin{lemma}\label{lemma: lower bound part I}
For any $m\in\N$ there exists $C_D>0$ such that
\[\lambda_1(X_w^0):=\inf_{\substack{u\in\Dom\eng_{X_w^0}\\u\neq 0}}\frac{\eng_{X_w^0}(u,u)}{\norm{u}^2_{L^2(X_w^0,\mu_{\vert_{X_w^0}})}}\leq \frac{C_D}{(sr)^m}\]
for all $w\in\Am$.
\end{lemma}
\begin{proof}
Consider $\nu\in\A^n$, $n\in\N$, such that $X_{\nu}\subseteq X^0$. Since $X^0$ is open and $X_{\nu}$ compact, we know that there exists a function $u\in\Dom\eng^0_{X}$ such that $u_{\vert_{X_{\nu}}}\equiv 1$, $u\geq 0$ and $\supp (u)\subseteq X^0$. 

Define
\begin{equation*}
u^w(x):=\left\{\begin{array}{rl} u\circ F_w^{-1}(x),& x\in X_w^0,\\
												0,&x\in \Xm ^0\setminus X_w^0.\end{array}\right.
\end{equation*}
Clearly $u^w\in\Dom\eng_{\Xm ^0}$ and analogously to the proof of Lemma~\ref{lemma: upper bound part I} we have by Corollary~\ref{corollary: energy identity} that
\begin{equation*}
\eng_{\Xm ^0}(u^w,u^w)=\sum_{w'\in\Am}r^{-m}\eng(u^w\circ F_{w'},u^w\circ F_{w'})+\sum_{k=0}^{m-1}r^{-k}\sum_{w'\in\A^k}\sum_{e\in J_1}\int_0^{\alpha}\abs{(u^w\circ F_{w'})'}^2dx.
\end{equation*}
Since $\supp (u^w)\subseteq \Xm ^0$, the last term of this sum equals zero and the definition of $u^w$ leads to
\begin{equation}\label{eq: self-type equality II}
\eng_{\Xm }(u^w,u^w)=r^{-m}\eng(u^w\circ F_w,u^w\circ F_w)=r^{-m}\eng(u,u).
\end{equation}

On the other hand, by definition of $\mu$ we have that
\begin{align}\label{eq: ellipticity}
\int_{X_w^0}\abs{u^w}^2d\mu (x)&=\int_{\Xm ^0}\abs{u}^2d\mu(F_w (y))\geq\int_{X_{\nu}}\abs{u}^2d\mu(F_w(y))\nonumber\\
&=\mu(F_w(X_{\nu}))\geq s^{\abs{\nu}}\mu(X_w)
\end{align}
Applying~\eqref{eq: self-type equality II} and~\eqref{eq: ellipticity} we finally obtain
\begin{align*}
\lambda_1(X_w^0)&\leq \frac{\eng_{\Xm ^0}(u^w,u^w)}{\norm{u^w}^2_{L^2(X_w^0,\mu_{\vert_{\Xm ^0}})}}\leq\frac{r^{-m}\eng(u,u)}{s^{\abs{\nu}}\mu(X_w)}=\frac{C_D}{(rs)^m},
\end{align*}
where $C_D:=\frac{\eng(u,u)}{r^{\abs{\nu}}}$ is independent of $w$.
\end{proof}

\begin{lemma}\label{lemma: lower bound part II}
There exists a constant $C'>0$ and $x_0>0$ such that
\begin{itemize}
\item[(i)] if $0<rs<\frac{1}{9}$, then
\[C' x^{\frac{1}{2}}\leq N_D(x),\]
with $C'=\widetilde{C}$ of Lemma~\ref{lemma: upper bound part II},
\item[(ii)] if $rs=\frac{1}{9}$, then
\[C'  x^{\frac{1}{2}}\log x\leq N_D(x),\]
\item[(iii)] if $\frac{1}{9}<rs<\frac{1}{6}$, then
\[C' x^{\frac{\log 3}{-\log (rs)}}\leq N_D(x)\]
\end{itemize}
for all $x\geq x_0$.
\end{lemma}
\begin{proof}
Analogously to the proof of Lemma~\ref{lemma: upper bound part II}, let $x_0:=4\pi^2s^2r^2$ and consider $x\geq x_0$. There exists $m\geq 1$ such that
\begin{equation*}\label{eq: 2.choice of m}
\frac{4\pi^2}{\alpha\beta(sr)^{m-4}}\leq x<\frac{4\pi^2}{\alpha\beta(sr)^{m-3}}.
\end{equation*}
By Lemma~\ref{lemma: lower bound part I}, we have that 
\[
\lambda_1(X_w^0)\leq \frac{C_D}{(sr)^{m}}\leq x
\]
and hence $N(x;\eng_{X_w^0},\Dom\eng_{X_w^0})\geq 1$ for all $w\in\Am$. It follows from Lemma~\ref{lemma: Dirichlet decomp} that 
\[N_D(x)\geq C_2x^{\frac{\ln 3}{-\ln (rs)}}\]
with $C_2=3^3\left(\frac{\alpha\beta}{4\pi^2}\right)^{-\frac{\ln 3}{\ln(rs)}}=C_1/3$. 

Analogous arguments as in Lemma~\ref{lemma: upper bound part II} together with Lemma~\ref{lemma: Dirichlet decomp} complete the proof. In the case when $1/9<rs<1/6$, the estimation of the geometric series leads to 
\begin{align*}
N(x;\eng_{I_m},\dom\eng^0_{I_m})&\geq \frac{3^{-3}\left(\frac{\alpha\beta}{4\pi^2}\right)^{-1/2}}{9rs-1}C_2x^{\frac{-\ln 9rs}{2\ln rs}+\frac{1}{2}}+O(x^{1/2})
\end{align*}
and finally
\[
N(x,\eng,\dom\eng)\geq C'x^{\frac{-\ln 3}{\ln (rs)}}+O(x^{1/2}),
\]
with $C'=2\min\left\{1,\left(\frac{4\pi^2 }{3^6\alpha\beta}\right)^{\frac{1}{2}}\cdot\frac{1}{9rs-1}\right\}C_2$.
\end{proof}

\begin{proof}[\textbf{Proof of Theorem~\ref{theorem: d_S X}}]
Since $\Dom\eng^0\subseteq\Dom\eng$ and $\eng^0=\eng\vert_{\Dom\eng^0}$, the minimax principle yields $N_D(x)\leq N_N(x)$ for all $x>0$. The statement follows directly from Lemma~\ref{lemma: upper bound part II} and Lemma~\ref{lemma: lower bound part II}.
\end{proof}

\section{Heat Kernel Estimates}\label{sec HKEs}
 
In this section we shall assume that $1/3<\alpha < 1$ and that $\haus$ is the restriction of the Hausdorff 1-measure to $X$. under these assumptions, the heat kernel with respect to the Hausdorff $1$-measure satisfies Gaussian heat estimates. Note that the measure of a set with respect to $\haus$ is the sum of the lengths of the line segments contained in that set. Thus,
\[
\haus(X) = \sum_{k=1}^\infty 3^k\paren{ \frac{1-\alpha}{2}}^k = \frac{3(1-\alpha)}{3\alpha -1}.
\]
\begin{proposition}\label{meas regularity}
There is a positive constant $C$ such that for any $x\in X$ and $t \leq \diam X$
\[
\frac1C t \leq  \haus(B_t(x)) \leq Ct,
\]
where $\diam X$ is the diameter of $X$ and $B_t(x)$ is the metric ball around $x$ with respect to $R$, $d_G$ or the Euclidean distance.
\end{proposition}

\begin{proof}
By Theorem \ref{lemR}~\eqref{i3}, all three metrics are equivalent, so proving the inequality for any of them proves it for all of them.
Take $B = B_t(x)$ to be the ball with respect to $d_G$ --- the geodesic metric. $\haus(B) \geq 2t$ for $t < \operatorname{diam} X$ because $\haus$ measures lengths.

Assuming $n \in \N$ is such that $r^n \leq t \leq r^{n-1}$, $B_t(x)$ intersects at most $3$ cells of scale $r^n$, i.e. 
\[
\#\set{w~|~F_w(X)\cap B_t(x)\neq \emptyset, |w| = n} \leq 3,
\]
and it intersects at most $4$ line segments not contained in these cells. Thus,
\[
\haus(B_t(x))) \leq 3r^n \mu(X) + 4t \leq (4 \mu(X) + 3)t.
\]
\end{proof}

\begin{theorem}\label{thm:H1 HKEs}
If $1/3 <  \alpha < 1$, then $(\eng,\dom\eng)$ is a Dirichlet form on $L^2(X,\haus)$, where $\haus$ is the Hausdorff 1-measure, and this Dirichlet form has a jointly continuous heat kernel $p(t,x,y)$.  If $d$ is either $d_G$ or the Euclidean distance, there are $c_1,c_2,c_3,$ and $c_4$ depending only on the choice of the metric so that $p$ satisfies the following Gaussian estimates
\[
{c_1}{t^{-1/2}} \exp\paren{-\frac{c_2d(x,y)^2}{t}} \leq p(t,x,y) \leq {c_3}{t^{-1/2}} \exp\paren{-\frac{c_4d(x,y)^2}{t}}.
\]
\end{theorem}

\begin{proof}
For $d= d_G$ this is a result of Theorem \ref{lemR}, Proposition~\ref{meas regularity}, the fact that $d_G$ is a geodesic metric and~\cite[Theorem 15.10]{Kig12}. Note that by~ \cite[Proposition 7.6]{Kig12} the (ACC) condition is satisfied for a local resistance form like $\eng$ on a compact space. Since $d_G$ and Eucildean distance are equivalent metrics, this implies the result for Euclidean distance as well.
\end{proof}

\section{Fractal quantum graphs}
\label{sec-frg}

In this section we present an abstract construction which resembles many topological, metric, resistance and energy 
 properties of the Hanoi fractal quantum graph. 

\begin{definition}
A compact metric space $(X,d)$ is called a \emph{fractal quantum graph} with length system $(\Phi_k,\ell_k)$ if there are positive lengths $\set{\ell_k}_{k=1}^\infty$ and a set of embeddings $\Phi_k:I_k := [0,\ell_k] \to X$  such that $\Phi_k|_{(0,\ell_k)}$ are local isometries with disjoint images, i.e. $\Phi_j((0,\ell_j)) \cap \Phi_k((0,\ell_k)) = \emptyset$ for $j\neq k$. $I_k$ is thus homeomorphic to $\Phi_k([0,\ell_k])$ with the subspace topology induced by $X$, and for any $x\in(0,\ell_k)$ there is $\eps>0$ such that if $|y-x| < \eps$, then $d(\Phi_l(x),\Phi_k(y)) = |x-y|$.

Further, we define
\[
J_n := \bigcup_{k=1}^n \Phi_k((0,\ell_k))
\]
to be the union of the image of the interiors of $I_k$ for $k\leq n$ and assume that $\cap_{n=1}^\infty J_n^c$ is a \emph{totally disconnected compact set}. Here, $J_n^c$ denotes the complement of $J_n$ in $X$. 
\end{definition}
 
If $(X,d)$ is a  fractal quantum graph with length system $(\Phi_k,\ell_k)$, we define the space of functions $\F_n$ to be the functions $f:X\to\R$ such that $f\circ \Phi_k \in H^1([0,\ell_k])$ for all $k\leq n$, and $f$  is locally constant on $J_n^c$. Here, locally constant means that any $x\in J_n^c$ has a neighborhood $U_x\subset J_n^c$ which is relatively open in $J_n^c$ and $f|_{U_x}$ is constant. 

It is elementary to show that $\F_n$ is a linear space and that $\F_m\subset \F_n$ for $m\leq n$. Denoting $\F_*:=\cup_{n=1}^\infty\F_n$, we define the bilinear form $\eng_*:\F_*\times \F_* \to \R$ for $f,g\in\F_n$ by
\[
\eng_*(f,g) = \sum_{k=1}^\infty\int_0^{\ell_k} (f\circ\Phi_k)'(t)(g\circ\Phi_k)'(t) \ dt.
\]
It is straightforward to see that   $\eng_*$ is non-negative definite and satisfies the Markov property as in (RF5). Also, $\eng_*(f,f) = 0$ only if $f$ is a constant function because if $f\in\F_n$ is not constant it must be non-constant on some $\Phi_k(I_k)$.   

The form $\eng_n$, which is the restriction of $\eng_*$ to $\F_n$, induces the following pseudo-metrics on $X$
\[
R_n(x,y) = \sup\set{\frac{|f(x)-f(y)|^2}{\eng_n(f,f)}~|~f\in\F_n,~\eng(f,f)\neq 0}.
\]
It follows from the literature on resistance forms that $R_n$ satisfies the triangle inequality although $R_n(x,y)$ may vanish for $x\neq y$. In fact, if $x$ and $y$ are in the same connected component of $J_n^c$, then $R_n(x,y) = 0$, but it follows from an argument similar to that in Theorem~\ref{lemR} (3) that if $x,y\in J_n, x\neq y$, then $R_n(x,y)>0$.

\begin{theorem}\label{thm:limit metric}
Suppose that a compact metric space $(X,d)$ is a fractal quantum graph with length system $(\Phi_k,\ell_k)$. Then the following statements are equivalent: 
\begin{enumerate}
\item\label{i1}$R_n(x,y)$ converges to a metric $R$ on $X$ with the same topology as $d$;
\item there is a resistance form $\eng$ on $X$ with resistance metric $R=\lim_{n\to\infty}R_n(x,y)$. This metric induces the same topology as $d$, $\F_* \subset \dom\eng$, $\eng(f,f) = \eng_*(f,f)$ for all $f\in\F_*$, and $\F_*$ is dense in $\dom\eng$ in the sense that for all $g\in \dom\eng$ there is $\set{f_i}_{i=1}^\infty\subset \F_*$ such that $\lim\limits_{i\to\infty}\eng(f_i-g)= 
0$.
\end{enumerate}
\end{theorem}
Note that, since $X$ is compact with respect to the effective resistance metric, if $\set{f_n}_{n=1}^\infty$ converges to $g$ in energy, this implies that there exists $\{\tilde f_n\}_{n=1}^\infty\subset\dom\eng$ such that $f_n-\tilde f_n$ is constant and $\tilde f_n$ converges to $g$ in energy, uniformly, and even in $\frac12$-H\"older convergence with respect to the effective resistance metric.

\begin{proof}[Proof of $(2){\implies}(1)$]
Assume there is $\eng$ on $X$ with resistance metric $R$ such that $\eng(f,f) = \eng_*(f,f)$ for $f\in \F_*$. Then,
\begin{align*}
R(x,y)&  = \sup\set{\frac{|f(x)-f(y)|^2}{\eng(f,f)}~|~\eng(f,f) \neq 0, ~f\in\F_*}
\\ & = \sup_n\sup\set{\frac{|f(x)-f(y)|^2}{\eng(f,f)}~|~\eng(f,f) \neq 0, ~f\in\F_n} 
\\ & = \lim_{n\to\infty} R_n(x,y).
\end{align*} 
The first equality above is because $\F_*$ is dense in $\dom\eng$, and the last equality is because $\F_m\subset \F_n$ for $m\leq n$ so $R_n$ is increasing in $n$.
\end{proof}
\begin{proof}[Proof of $(1){\implies}(2)$]
Assume that $R(x,y) := \lim_{n\to\infty} R_n(x,y)$ is a metric on $X$ that induces the same topology on $X$ as the metric $d$. In this case, 
\begin{align}
\label{engbound}
|f(x)-f(y)|^2 \leq \eng_n(f,f)R_n(x,y) \leq \eng_n(f,f)R(x,y)
\end{align}
for any $f\in \F_n$, and thus $\F_n\subset C(X)$, where $C(X)$ is the set of continuous functions on $X$ (note there is no ambiguity in $C(X)$ because $d$ and $R$ are assumed to induce the same topologies).

This implies that $\Phi_k((0,\ell_k))$ is an open set in $X$ for any $k$ because $\Phi_k((0,\ell_k)) = f^{-1}((0,\infty))$, where $f$ is the function in $\F_n$, $n\geq k$, defined to be $0$ on the complement of $\Phi_k((0,\ell_k))$ and satisfying $f\circ \Phi_k(t) = t(\ell_k-t)$ for $t\in I_k$.

For any finite subset $S\subset X$, define
\[
\eng_S(f,f):= \inf\set{\eng_*(u,u)~|~u\in\F_*,~u|_S = f}.
\]
We establish that $\eng_S$ is a resistance form on $\ell(S)$ proving first that $\eng_S(f,f)$ is well defined for all $f\in \ell(S)$. Let us consider $x,y\in S$. Since $R(x,y) > 0$, there is $n\in\N$ such that $R_n(x,y)>0$ and therefore $u\in\F_n$ such that $u(x)\neq u(y)$, so that $\F_*$ separates points in $S$. Since $u+c\in\F_n$ for any $c\in\R$, $\F_*$ vanishes nowhere on $S$. This implies that $\eng_S$ has domain $\ell(S)$.

\begin{enumerate}
[(RF1)]
\item $\eng_*$ is symmetric and non-negative definite, so $\eng_S$ must be as well. If $f\equiv c$ is a constant function, then $u\equiv c\in\F_*$ and $\eng_*(u,u) = 0$, hence $\eng_S(f,f)= 0$. On the the other hand, if $f\in\ell(S)$ is non-constant, then there are $x,y\in S$ with $x\neq y$ and $f(x)\neq f(y)$ so that
\[
\eng_*(u,u) \geq \frac{|f(x)-f(y)|^2 }{R(x,y)} > 0. 
\] 
\item This follows from (R1) because $\ell(S)$ is finite-dimensional. 
\item Since the domain of $\eng_S$ is $\ell(S)$, there is clearly $f\in\ell(S)$ such that $f(x)\neq f(y)$.
\item This follows from the bound in (\ref{engbound}) and a similar argument to that in (RF4) in Subsection~\ref{subsubsec:resistanceform}.
\item  $\eng_*$ has the Markov property, which implies that $\eng_S$ has the Markov property as well. 
\end{enumerate}

Next, we select a sequence of finite sets $\set{S_n}_{n=1}^\infty$ with $S_m\subset S_n$ for $m\leq n$ and such that $S_* :=\cup_{n=1}^\infty S_n$ is dense in $X$. It follows from the argument in 
Subsection~\ref{subsubsec:resistanceform} of the proof of Theorem~\ref{thm resistance form} 
that $\eng_{S_n}$ is a compatible sequence. Thus we may apply~\cite[Theorem 3.13]{Kig12} to obtain a resistance form 
\[
\eng(f,f) = \lim_{n\to\infty}\eng_{S_n}(f|_{S_n},f|_{S_n}),
  \ \qquad
\dom\eng = \set{f\in \ell(S_*)~|~\lim_{n\to\infty}\eng_{S_n}(f|_{S_n},f|_{S_n})<\infty}.
\]
If $R_\eng$ and $R_{\eng_{S_n}}$ are the resistance metrics associated to $\eng$ and $\eng_{S_n}$ respectively, then we have that $R_\eng(x,y) = R_{\eng_{S_n}}(x,y)$ for all $x,y\in S_n$. From the calculation in (RF4) of Subsection~\ref{subsubsec:resistanceform} we have that $R_{\eng_{S_n}}(x,y) = R(x,y)$ for all $x,y\in S_n$. Thus $R_\eng(x,y) = R(x,y)$ for all $x,y\in S_*$ and since $X$ is compact (and hence complete), $(X,R)$ is isometric to the completion $S_*$ with respect to $R_\eng$. In view of~\cite[Theorem 3.14]{Kig12}, we get
\begin{align}
\label{dom}
\dom\eng = \set{f\in C(X)~|~\lim_{n\to\infty}\eng_{S_n}(f|_{S_n},f|_{S_n})<\infty}.
\end{align}

To see that $\F_n\subset \dom\eng$, notice that $\eng_{S_m}(f|_{S_m},f|_{S_m}) \leq \eng_n(f,f)$ for any $m$ and $f\in\F_n$, and hence $\eng(f,f) \leq \eng_n(f,f)$.

To see that $\eng(f,f) = \eng_{n_0}(f,f)$ for any $f\in\F_{n_0}$, we assume without loss of generality that for any $n$, $S_n$ is a $1/n$-net, i.e. for any $x\in X$ there is $y\in S_n$ with $R(x,y)\leq 1/n$, and $\Phi_k(0),\Phi_k(\ell_k)\in S_n$ for all $k\leq n$. In this situation, for $n\geq n_0$, the minimum  of $\eng_*(u,u)$ with $u|_{S_n} = f|_{S_n}$ is attained by the function $u_n$ such that $u_n\circ \Phi_k(t) = f\circ \Phi_k(t)$ for $t\in \Phi_k^{-1}(S_n)$ and is linear everywhere else on $\Phi_k((0,\ell_k))$, as this minimizes energy on $I_k$ for all $k\leq n$. Finally, set $u_n(x) = f(x)$ for $x\in J_{n_0}^c$, i.e. extend to the rest of $J_{n_0}^c$ by constants.  Notice that for this $u_n$ to be well defined, it is important that we assumed that $\Phi_k(0)$ and $\Phi_k(\ell_k) \in S_{n_0}$ if $k\leq n_0$, because $\Phi_k(0)$ and $\Phi_k(\ell_k)$ are in both $\Phi_k(I_k)$ and $J_{n_0}^c$. In particular, $u\in\F_{n_0}$.  
We have established that 
\[
\eng(f,f) = \lim_{n\to\infty}\eng_{S_n}(f|_{S_n},f|_{S_n}) =\lim_{n\to\infty}\eng_{n_0}(u_n,u_n) = \eng_{n_0}(f,f),
\]
where the last inequality holds because $\Phi_k^{-1}(S_n)$ becomes uniformly dense in $I_k$ and thus
\[
\lim_{n\to\infty} \int_0^{\ell_k} ((u_n\circ\Phi_k)'(t))^2 \ dt = \int_0^{\ell_k} ((f\circ\Phi_k)'(t))^2 \ dt
\]
for all $k\leq n_0$.

To see that $\F_*$ is dense in $\dom\eng$, 
 we know from the definition of the domain in~\eqref{dom} that for any $f \in \dom\eng$ there is $f_n\in\dom\eng$ such that $\eng(f_n,f_n) = \eng_{S_n}(f_n|_{S_n},f_n|_{S_n})$ and $f_n(x) = f(x)$ for all $x\in S_n$. In particular $\lim_{n\to\infty}\eng(f_n,f_n) = \eng(f,f)$.
Since 
\[
\eng_{S_n}(f_n|_{S_n},f_n|_{S_n})  =\inf\set{\eng_*(u,u)~|~u\in\F_*,~u|_{S_n}=f|_{S_n}},
\]
there is $(u_{n,k})_{k\in\N}\subset \F_*$ with $\lim_{k\to\infty}\eng_*(u_{n,k},u_{n,k}) =\lim_{k\to\infty}\eng(u_{n,k},u_{n,k})= \eng_{S_n}(f_n|_{S_n},f_n|_{S_n})$ and $u_{n,k}(x) = f_n(x) = f(x)$ for all $x\in S_n$. By diagonalizing and passing to a subsequence if required, $\lim_{n\to\infty}\eng(u_{n,n},u_{n,n}) = \eng(f,f)$. Since $u_{n,n}(x) = f(x)$ on the $1/n$-net $S_n$, for any $y\in X$ and arbitrary $x\in S_n$ we have
\[
|u_{n,n}(y)-f(y)| = |(u_{n,n}(y)-f(y)-(u_{n,n}(x) - f(x))| \leq R(x,y)\eng(u_{n,n}-f,u_{n,n}-f). 
\] 
This quantity vanishes as $n\to\infty$, which establishes that $\F_*$ is dense in the prescribed manner.
\end{proof}

\begin{definition} 
A fractal quantum graph $(X,d)$ is called a \emph{proper  fractal quantum graph} if the maps $\Phi_k: (0,\ell_k) \to X$ are open. 

A compact geodesic metric space $(X,d)$ which is a proper quantum graph is called a \emph{proper geodesic fractal quantum graph}.
\end{definition}

This definition means that $I_k$ is homeomorphic to $\Phi_k([0,\ell_k])$ with the subspace topology induced by $X$, and for any $x\in(0,\ell_k)$, there is $\eps>0$ such that the $\eps$-neighborhood of $x$ is mapped isometrically onto the $\eps$-neighborhood of $\Phi_k(x)$. In particular, $|y-x| < \eps$ if and only if $d(\Phi_k(x),\Phi_k(y)) = |x-y|< \eps$.  

Note that the assumption of the existence of the geodesic metric for a local resistance form is natural because of the results in~\cite{HKT13}. With this assumption, we have the following theorem. 

\begin{theorem}\label{thm:pgfqg}
Any proper geodesic fractal quantum graph satisfies conditions of Theorem~\ref{thm:limit metric}(\ref{i1}) and so $\eng_*$ extends to the resistance form $\eng$ on $X$ with resistance metric $R(x,y)=\lim\limits_{n\to\infty}R_n(x,y)$, which induces the same topology as $d$.
\end{theorem}


\begin{proof}
It is easy to show that $R_n(x,y)\leqslant d(x,y)$ by the same method as in Lemma~\ref{engestimate}. Thus, since $R_n(x,y)$ is an increasing and bounded sequence, it must converge to $R(x,y)\leqslant d(x,y)$. Since $R_n$ satisfies the triangle inequality, so must $R$, and all that is left to establish that $R$ is a metric is to show that $R(x,y) >0$ when $x\neq y$.
In fact, $J_n^c$ are compact subsets with totally disconnected intersection, and so any distinct $x,y\in X$ can be separated by two disjoint compact subsets by removing finitely many open edges. Thus there is $n\in\N$ such that $R_n(x,y) >0$. This settles all topological questions. For instance, the common base of open sets, both for $d$ and $R$, can be defined as follows:  all open subsets of the open edges $\Phi_k((0,\ell_k))$; for $\eps_n$ small enough, connected components of $\eps_n$-neighborhoods of $J_n^c$ with respect to the  metric $d$.  
The notion of small enough $\eps_n$ is understood in the sense that all these connected components of $\eps_n$-neighborhoods of $J_n^c$ should have either no intersection, or be contained one in another.  
\end{proof}

\begin{remark}Note that  the closed edges $\Phi_k(I_k)$ together with the complements  $J_n^c$ will define a finitely ramified cell structure, in the sense of \cite{Tep08}.
\end{remark}

\begin{remark}
Note that we do {\bf not} claim in Theorems~\ref{thm:limit metric} and \ref{thm:pgfqg} that the domain of $\eng$ can be identified with an analogue of $H^1(X)$. This was one of the main results for the Hanoi fractal quantum graph, which may require extra assumptions in a more general situation. 
\end{remark}

\begin{example}
Since the Hanoi quantum graph provides a good example of a proper geodesic quantum graph, we also would like to present as standard counterexample the infinite broom (see, for instance, \cite{SS95}): Let $X:= \cup_{k=0}^\infty \Phi_k([0,\ell_k])$, where $\Phi_k:[0,\ell_k] \to \R^2$ are defined by $\ell_0 = 1$, $\Phi_0(t) = (t,0)$, $\ell_k = \sqrt{1+k^{-2}}$ and 
\[
\Phi_k(t) = \frac{t}{\sqrt{1+k^{-2}}}(1,k^{-1}).
\]
If we equip $X$ with the Euclidean distance, $X$ along with the maps $\set{\Phi_k}_k$ form a compact fractal quantum graph that is not proper. In particular, the functions in $\F_*$ are not necessarily continuous, for example the function such that $f\circ\Phi_1(t)  = t$ and $f(x) = 0$ for $x\notin \Phi((0,1])$. Thus $R_n$ cannot converge to a metric which induces the same topology. However, $R_n$ does converge to a geodesic metric $R$ on $X$. With this metric, $X$ is isometric to the space $\sqcup_{k=1}^\infty I_k/\sim$ where $\sim$ is the equivalence relation that identifies the $0$ element in each $I$, and $R$ is the length metric induced by Euclidean distance. Thus $\eng_*$ induces a resistance form on this metric space. However this space is not compact in the effective resistance topology, and not even locally compact. Many related questions are discussed in \cite{Kig95,Kig12}. 
\end{example}

\section{Generalized Hanoi-type quantum graphs}
In this section we briefly present a multidimensional version of the 
Hanoi quantum graphs. 
Let $N_0>2$ be a natural number and let $\alpha>0$ be fixed. Further, consider the alphabet $\A_{N_0}:=\{1,\ldots ,N_0\}$ and the contractions $F_i\colon\R^{N_0-1}\to\R^{N_0-1}$, $i\in\A_{N_0}$. Each mapping $F_i$ has contraction ratio $r_i=\frac{1-\alpha}{2}$($<1$) and fixed point $p_i$. We also set $V_{N_0}:=\{p_1,\ldots ,p_{N_0}\}$. 

The \textit{generalized Hanoi attractor of parameters} $N_0$ and $\alpha$ is the unique non-empty compact subset of $\R^{N_0-1}$ such that 
\[K_{\alpha,N_0}=\bigcup_{i=1}^{N_0}F_i(K_{\alpha,N_0})\cup\bigcup_{\{i,j\}\subset\A_{N_0}^2}[i,j],\]
where $[i,j]$ denotes the straight line joining the points $F_i(p_j)$ and $F_j(p_i)$ (note that $i~\neq~j$). It is easy to see that the Hausdorff dimension of this set is given by 
\[
\max\left\{1,\frac{\ln N_0}{\ln 2-\ln(1-\alpha)}\right\}.
\]
If we choose $\alpha$ in the interval $(0,\frac{N_0-2}{N_0})$, then $\dim K_{\alpha,N_0}>1$ and we obtain a fractal. In the following, we will only consider $\alpha$ belonging to this interval.

\begin{remark}
The case $N_0=3$ corresponds to the Hanoi attractor treated in Sections 3-5. In the case $N_0=4$, $K_{\alpha,N_0}$  fits into a tetrahedron of side length $1$.
\end{remark}

Let us now consider the generalized Hanoi attractor of parameter $N_0$ for a fixed $\alpha$ and denote it by $X_{N_0}$. This set may be approximated by the sequence of metric graphs $(X_{N_0,n})_{n\in\N}$, where {$X_{N_0,n}:=(V_{N_0,n},E_{N_0,n},\partial,r)$} is defined analogously to Definition~\ref{def ApproxMetricGraphs} just substituting $\A$ by $\A_{N_0}$.

By doing the obvious substitutions in Definition~\ref{def LevelEnergy}, we define the energy of the $n-$th approximation of $X_{N_0}$, $\eng_{N_0,n}\colon\mathcal{F}_{N_0,n}\times\mathcal{F}_{N_0,n}\to\R$ by
\[
\eng_{N_0,n}(u,v):=\int_{X_{N_0}}u'v'\,dx
\]
for all $u,v\in\F_{N_0,n}$, i.e. functions everywhere constant out of finitely many segments corresponding to ``joining-type'' edges of $X_{N_0,n}$.
By the same arguments as in Section 5 we get a suitable domain $\Dom\eng_{N_0}$ on $X_{N_0}$ such that


\begin{proposition}
$(\eng_{N_0},\F_{N_0})$ is a resistance form.
\end{proposition}

From this resistance form, we obtain a Dirichlet form $(\eng_{N_0},\dom\eng_{N_0})$ by considering a measure $\mu_{N_0}$ on $X_{N_0}$ following the construction of $\mu$ in Section~\ref{sec-asymp}. We thus introduce the parameter $\beta~>~0$ that measures the lines of length $\alpha$. This parameter needs to belong to the interval $\big(0,\frac{2}{N_0(N_0-1)}\big)$ because otherwise, since
\[
s:=\mu_{N_0}(X'_{N_0})=\frac{2-N_0(N_0-1)\beta}{2N_0},
\]
where $X'_{N_0}$ denotes any first-level copy of $X$, would be zero or negative. 

The definition of $s$ comes from the fact that we want the measure $\mu_{N_0}$ to satisfy
\[1=\mu(X)=\frac{N_0(N_0-1)}{2}\beta+N_0\mu_{N_0} (X_{N_0}'),\]
where $\frac{N_0(N_0-1)}{2}$ is the number of straight lines joining the different copies $X'_{N_0}$.

Following the proofs of Section \ref{sec-asymp} just replacing $X$ by $X_{N_0}$ and $(\eng,\Dom\eng)$ by $(\eng_{N_0},\Dom\eng_{N_0})$, one obtains Theorem~\ref{theorem: d_S X_N_0} on the spectral asymptotics of the corresponding eigenvalue counting function of the associated Laplacian, leading to the spectral dimension of $X_{N_0}$. In this more general case, it follows directly from the choice of $\alpha$ and $\beta$ that 
\[
rs={(1-\alpha)[2-\beta N_0(N_0-1)]}/{2N_0}<{1}/{2N_0}.
\]
Finally, using the techniques from Section~\ref{sec HKEs}, if $\alpha \in \big(\frac{N_0-2}{N_0},1\big)$, then the Dirichlet form with respect to the 1-Hausdorff measure has a jointly continuous heat kernel which satisfies Gaussian estimates of the form given in Theorem~\ref{thm:H1 HKEs} with respect to either the geodesic metric or the Euclidean metric. 
%
\bibliography{Fractals}{}
\bibliographystyle{amsalpha}
%
\end{document}